\documentclass[12pt]{article}

\usepackage{indentfirst}  
\usepackage{amsfonts,amssymb,amsmath}
\usepackage{bm}         
\usepackage{cite}
\usepackage{url}
\usepackage{enumitem}
\usepackage[dvipsnames]{xcolor} 
\usepackage{MnSymbol}
\usepackage{hyperref}
\usepackage{tensor}
\usepackage{arcs}
\usepackage{tikz-cd}
\usepackage[titletoc,title]{appendix} 
\usepackage{tensor}
\usepackage{upgreek}
\usepackage{titlesec}

 \titlespacing*{\section}{0pt}{3.5ex plus 1ex minus .2ex}{1.2ex plus .2ex}
\usepackage{etoolbox}

\gappto\normalsize{%
  \setlength{\abovedisplayskip}{5pt plus 2pt minus 3pt}%
  \setlength{\belowdisplayskip}{5pt plus 2pt minus 3pt}%
  \setlength{\abovedisplayshortskip}{0pt plus 2pt}%
  \setlength{\belowdisplayshortskip}{3pt plus 1pt minus 1pt}%
}

\usetikzlibrary{arrows}

\AtBeginDocument{}%

\newcommand{\corurl}{red}
\newcommand{\corcite}{ForestGreen}
\newcommand{\corlink}{blue}
\definecolor{dgreen}{cmyk}{1,0,1,0.25}
\newcommand{\dd}{\mathrm{d\!\raisebox{-0.05pt}{\scalebox{.8}[1.0]{I}}}}

\newcommand{\bd}{\mathrm{d}}
\newcommand{\iprod}{\raisebox{8pt}{\scalebox{1}[-2.5]{$\neg\,$}}}

\newcommand{\ee}{\mathsf{e}}

\newcommand{\ddd}{\mathrm{d}}

\newcommand{\Aa}{\mathsf{A}} 
\newcommand{\uu}{\mathsf{u}} 
\newcommand{\iiiprod}{\iprod\!\!\!\!\iprod}

\let\oldiprod\iprod
\renewcommand{\iprod}{\mathbin{\smash{\oldiprod}}}
 \allowdisplaybreaks
\newcommand{\stackmark}[3]{%
  \mathbin{%
    \mathchoice%
      {\ooalign{$\displaystyle#1$\cr\hidewidth\smash{\raisebox{#3}{$\scriptstyle#2$}}\hidewidth\cr}}%
      {\ooalign{$\textstyle#1$\cr\hidewidth\smash{\raisebox{#3}{$\scriptstyle#2$}}\hidewidth\cr}}%
      {\ooalign{$\scriptstyle#1$\cr\hidewidth\smash{\raisebox{#3}{$\scriptscriptstyle#2$}}\hidewidth\cr}}%
      {\ooalign{$\scriptscriptstyle#1$\cr\hidewidth\smash{\raisebox{#3}{$\scriptscriptstyle#2$}}\hidewidth\cr}}%
  }%
}

\newcommand{\doublewedgebase}{%
  \raisebox{0ex}{\scalebox{0.75}{$\wedge$}}\mkern-10mu\wedge%
}

\newcommand{\wed}{\stackmark{\wedge}{\cdot}{-0.2ex}}
 
\newcommand{\wwed}{%
  \mathbin{%
    \mathchoice%
      {\ooalign{$\displaystyle\doublewedgebase$\cr\hidewidth\raisebox{-0.2ex}{\kern-0.1em$\scriptstyle\cdot$}\hidewidth\cr}}%
      {\ooalign{$\textstyle\doublewedgebase$\cr\hidewidth\raisebox{-0.2ex}{\kern-0.21em$\scriptstyle\cdot$}\hidewidth\cr}}%
      {\ooalign{$\scriptstyle\doublewedgebase$\cr\hidewidth\raisebox{-0.2ex}{\kern-0.05em$\scriptscriptstyle\cdot$}\hidewidth\cr}}%
      {\ooalign{$\scriptscriptstyle\doublewedgebase$\cr\hidewidth\raisebox{-0.2ex}{\kern-0.05em$\scriptscriptstyle\cdot$}\hidewidth\cr}}%
  }%
}

\newcommand{\iproddot}{%
  \mathbin{%
    \mathchoice%
      {\ooalign{$\displaystyle\iprod$\cr\hidewidth\raisebox{-0.8ex}{\kern-0.05em$\scriptstyle\cdot$}\hidewidth\cr}}%
      {\ooalign{$\textstyle\iprod$\cr\hidewidth\raisebox{-0.8ex}{\kern-0.05em$\scriptstyle\cdot$}\hidewidth\cr}}%
      {\ooalign{$\scriptstyle\iprod$\cr\hidewidth\raisebox{-0.6ex}{\kern-0.03em$\scriptscriptstyle\cdot$}\hidewidth\cr}}%
      {\ooalign{$\scriptscriptstyle\iprod$\cr\hidewidth\raisebox{-0.5ex}{\kern-0.03em$\scriptscriptstyle\cdot$}\hidewidth\cr}}%
  }%
}

\hypersetup{linktocpage,colorlinks,urlcolor=\corurl,citecolor=\corcite,linkcolor=\corlink,
pdftitle={Hamiltonian analysis of Carrollian Gravity with the Immirzi parameter},
pdfauthor={J.F. Barbero, J. Margalef-Bentabol, V. Varo, E.J.S. Villase\~nor}}

\numberwithin{equation}{section}  

\newtheorem{theorem}{Theorem}[section]

\def\QED{{\boldmath$\rule{0.5em}{0.5em}$}}                                
\def\markatright#1{\leavevmode\unskip\nobreak\quad\hspace*{\fill}{#1}}    
\def\qed{\markatright{\QED}}                                              

\newtheorem{lemma}[theorem]{Lemma}
\newenvironment{proof}[1][Proof]{\noindent\textbf{#1.} }{\qed\\}
\newtheorem{proposition}[theorem]{Proposition}

\usepackage{titling}
\usepackage{authblk}
\title{From Holst to Carroll Gravity, \\ a Hamiltonian point of view}

\author[1]{J. Fernando Barbero G.}
\author[2]{Juan Margalef-Bentabol\,}
\author[3,4]{Valle Varo\,}
\author[3,4]{Eduardo J.S. Villase\~nor\,}
\affil[1]{\href{https://ror.org/05rtchs68}{Instituto de Estructura de la Materia}, IEM-CSIC, Serrano 123, 28006 Madrid, Spain}
\affil[2]{Department of Mathematics and Statistics \href{https://ror.org/0161xgx34}{Universit\'e de Montr\'eal}, Montr\'eal, Qu\'ebec, Canada}
\affil[3]{\href{https://ror.org/03ths8210}{Universidad Carlos III de Madrid}, Departamento de Matem\'aticas, Avenida de la Universidad, 30
 (edificio Sabatini), 28911, Legan\'es (Madrid), Espa\~na}
\affil[4]{Grupo de Teor\'{\i}as de Campos y F\'{\i}sica Estad\'{\i}stica. Instituto Gregorio Mill\'an (UC3M). Unidad Asociada al Instituto de Estructura de la Materia, CSIC}
\date{}                     
\begin{document}
	
\maketitle
\date{June 30, 2026}

\vspace*{-8ex}

\begin{abstract}
The Carrollian regime of gravity provides a useful ultrarelativistic framework for studying asymptotically flat spacetimes, horizon dynamics, and condensed matter systems. In this paper, we present a comprehensive Hamiltonian analysis of the most general Carroll-invariant Lagrangian that can be derived from the Holst action. To guide the analysis, Cartan geometry tools are used. We identify the full constraint structure of the theory, characterize its gauge symmetries, obtain the explicit form of the Hamiltonian vector fields and introduce Ashtekar-like variables for the magnetic Carrollian regime. We also discuss in detail an analog of the time gauge employed in the Hamiltonian analysis of the Holst action for general relativity. 

\end{abstract}

%


\noindent
{\bf Key Words:}
Carrollian gravity; Group contractions; Gravitational actions.

%
%
	
\section{Introduction}\label{sec_Carrol_Introduction}

The Carrollian limit of gravitational theories provides an interesting setting to study a number of problems related to general relativity (GR) and cosmology such as the study of black hole horizons \cite{Donnay:2019jiz} and asymptotic null infinity \cite{D1}. It also finds interesting applications in other parts of physics \cite{figueroa2023carroll}. Many of these applications take advantage of the fact that in the Carrollian limit light cones degenerate into curves and, hence, the nature of causality in Carrollian spacetimes is somewhat simplified (although, arguably, it is changed in a rather non-trivial fashion). From a mathematical perspective, the main idea in the Carrollian setting is to introduce a contraction of the Poincar\'e algebra  within the gravitational models at hand along the lines first proposed by L\'evy-Leblond \cite{Levy-Leblond} and Sen Gupta \cite{sen_gupta}. In recent years there has been a renewed interest to understand the intricacies of the dynamics of the electric and magnetic sectors of Carrollian theories \cite{Campoleoni:2022ebj} and extend their applications (see recent papers, such as \cite{Ecker:2024czx, March:2024zck, deBoer:2023fnj, Bergshoeff:2025dhy, Henneaux:2021yzg}, for comprehensive lists of references). 

The purpose of this paper is to study in detail the Hamiltonian formulation for the most general Carrollian gravitational Lagrangian derivable from the Holst action  \cite{Holst:1995pc} for GR. This Lagrangian is the sum of the Carrollian gauge-invariant 4-forms given in \cite{Figueroa-OFarrill:2022mcy} (see also \cite{Bergshoeff:2017btm}). One application of the results that we present here is to provide relevant toy models for the canonical quantization of gravity. In fact, for the action discussed in this work, we will find that the phase space is related to the one corresponding to the Ashtekar formulation for GR \cite{Ashtekar:1986yd, Ashtekar:1987gu}. The only significant difference between the Hamiltonian formulation presented here and the standard Hamiltonian formulation for Lorentzian gravity in terms of real Ashtekar variables is the form of the Hamiltonian constraint: whereas in the case of GR this constraint can be written as the combination of two terms, one involving the curvature of the Ashtekar connection and the other the curvature of the spin connection built with the triads, in the Carrollian case only the last term is present. This is a simplification of sorts that may enable the canonical quantization of the action discussed here.

We would like to point out that the present work is not the first one to study the Hamiltonian formulation of a Carrollian gravity. For example, the equivalence between the magnetic Carrollian limit of Einstein gravity and the Carrollian theory of gravity was established in \cite{Campoleoni:2022ebj} by using Hamiltonian techniques.  Also, in \cite{Sengupta:2022rbd} the author performs a Hamiltonian analysis of the Carrollian limit of the Cartan-Palatini action for GR the action discussed in \cite{Bergshoeff:2017btm} by following a line of thought similar to the one presented in \cite{BarroseSa:2000vx} (in particular, relying on the use of a quadratic constraint) and taking a limit. There are several well-established methods to deal with singular (constrained) Lagrangian systems and obtain the corresponding Hamiltonian formulations. The best known of them is due to Dirac \cite{Dirac}. If followed to the letter, it always gives a Hamiltonian picture of dynamics equivalent to the one given from the action through the Euler-Lagrange equations. However, the implementation of the procedure is often tricky, in particular, the resolution of the equations that determine the undetermined multipliers introduced in the definition of the total Hamiltonian. If these equations are not solved, it is possible to miss constraints and arrive at incorrect results. A nagging feature of the approach followed in \cite{BarroseSa:2000vx} is the substitution of primary constraints linear in momenta by a quadratic one which is then used to define the total Hamiltonian. The reason for this is rather obscure but is probably justified by the appearance of dynamical objects in the primary constraints that are interpreted as Lagrange multipliers and, for this reason, considered (incorrectly) as non-dynamical objects. Another unnatural feature of the mentioned approach is that, instead of working with the natural variables in which the Cartan-Palatini action is written (a coframe and an spin connection) it uses a frame field. This unnecessarily complicates the form of the action and makes all the computations significantly harder.

The approach we follow here is different in spirit on several grounds. First, because we consider the most general Carrollian Lagrangian derivable from the Holst action (see \cite{Figueroa-OFarrill:2022mcy}), including a term with the Immirzi parameter. Second, from a methodological standpoint, we follow the clean and geometrically inspired approach developed by Gotay, Nester and Hinds (GNH) in \cite{Gotay1,Gotay2} to study constrained systems. This method provides an alternative to Dirac's approach, while yielding essentially equivalent information. Finally, it offers several computational and conceptual advantages. 

To complete the analysis of the relevant equations in the GNH setting (and, likewise, in Dirac's formalism), it is often useful to exploit geometric insights that clarify the structure of their solutions. In particular, to solve the equations that determine the components of the Hamiltonian vector fields, brute-force computations can always be used. However, such an approach frequently obscures the underlying structure of the solutions, making their interpretation difficult. It can also be challenging to verify the geometric consistency conditions that these vector fields must satisfy. A preliminary geometric analysis, together with the identification of the relevant symmetries, provides a clearer picture of the expected form of the Hamiltonian vector fields. This allows one to understand their structure \textit{a priori} and considerably simplifies the resolution of the corresponding equations. For the present paper, the geometric insight provided by Cartan geometry (see \cite{catren2015geometric,sharpe2000differential}) is especially useful and helpful.

The paper is organized as follows. In Section~\ref{sec_geometric}, we review the aspects of Cartan geometry that provide the geometric framework and motivation for our analysis. Section~\ref{sec_Action} introduces the action that we study in the paper, discusses its symmetries and field equations, and derives several immediate consequences for the dynamics. As we will show, a number of these features are characteristic of Carrollian systems. In Section~\ref{sec_Hamiltonian}, we develop the Hamiltonian formulation of the model using the GNH method. We identify the constraint structure, solve the equations defining the Hamiltonian vector fields, and establish the consistency of the resulting dynamics. Section~\ref{sec_time_gauge} is devoted to the implementation of the time gauge, which leads to a particularly simple set of canonical variables and a simplified description of the dynamics. We conclude in Section~\ref{sec_conclusions} with a discussion of our results and possible directions for future work. Appendix~\ref{sec_math_results} lists a number of useful mathematical facts that we use throughout the paper and summarizes our notation, and in Appendix~\ref{app_juan4} we discuss the transformations of the spacetime fields under a spacetime vector field decomposed relative to a spatial hypersurface.

%
%

\section{Geometric arena}\label{sec_geometric}

In this section, we describe Carrollian spacetimes within Cartan geometry since, as in GR, it provides the most convenient setting to work with connection-type variables.

A \((G,H)\)-Cartan geometry on a smooth manifold \(M\) is given by an \(H\)-principal bundle \(P\) endowed with a Cartan connection \(\omega\in\Omega^{1}(P,\mathfrak g)\), i.e.,\ a \(\mathfrak g\)-valued 1-form satisfying the following conditions (see \cite{sharpe2000differential,vcap2024parabolic,catren2015geometric} for further details):
\begin{enumerate}
    \item for every \(p\in P\), the map \(
\omega_p:T_pP\to\mathfrak g
\)  is a linear isomorphism; 
\item \(
(R^h)^*\omega=\mathrm{Ad}(h^{-1})\circ\omega,
\)
for all \(h\in H\), where \(R^h\) is the right action of \(h\) on \(P\);
\item \( \omega(X^*)=X,\)
for all \(X\in\mathfrak h\subset\mathfrak g\), where \(X^*\) is the fundamental vector field associated with \(X\).
\end{enumerate} 

This paper focuses on the Hamiltonian formalism, therefore we will restrict ourselves to trivial (and trivialized) principal bundles
\[
H \rightarrow P=M\times H \rightarrow M,
\]
equipped with the canonical section
\(s_e:M\to P\),  \(\, s_e(x)=(x,e)\),
where \(e\) denotes the identity element of \(H\). This choice allows us to pull back the Cartan connection to a \(\mathfrak g\)-valued 1-form on \(M\),
\[
W=s_e^*\omega\in\Omega^1(M,\mathfrak g).
\]
A different choice of section,
\(s_h(x)=(x,h(x))\),
with \(h:M\to H\) smooth, induces the gauge transformation 
\[
W\mapsto \mathrm{Ad}(h^{-1})\circ W+h^{-1}\ddd h.
\] 
The (vertical) automorphisms of a trivial bundle are $\mathrm{Gau}(M\times H)= C^{\infty}(M,H)$.

For the purposes of this paper, it is enough to assume that \(\mathfrak g\) is reductive, namely
\[
\mathfrak g=\mathfrak h\oplus\mathfrak m,
\]
which is an \(\mathrm{Ad}(H)\)-invariant decomposition. In particular, \(W\) decomposes as
\[
W=W_{\mathfrak h}+W_{\mathfrak m}.
\]
Under this decomposition, \(W_{\mathfrak{m}}\) and \(W_{\mathfrak{h}}\) respectively correspond to the solder form and the Ehresmann connection (pulled back to \(M\)).
Notice that, by the defining properties of a Cartan connection, the map
\begin{equation}\label{eq: W_m iso}
(W_{\mathfrak m})_x:T_xM\rightarrow\mathfrak m
\end{equation}
is a linear isomorphism for every \(x\in M\). Consequently, \(M\) must be parallelizable. This topological restriction is already implicit in the assumption that the principal bundle \(P\) is trivial. One could avoid this restriction by working with an open cover \(\{U_\alpha\}_\alpha\) of \(M\) by contractible open sets and gluing the corresponding local Cartan geometries using smooth transition functions \(h_{\alpha_1\alpha_2}\). We shall not consider this more general situation here because, as discussed in the next section, all the manifolds used in the paper are parallelizable.

In what follows, $
G=\mathrm{Carr}(1+3)=\mathrm{SO}(3)\ltimes \bigl(\mathbb{R}^3\ltimes (\mathbb{R}\times \mathbb{R}^3)\bigr)$
is the Carroll group and $H$ the closed subgroup  $
H\cong\mathrm{SO}(3)\times \mathbb{R}^3$ generated by the rotations and Carrollian boosts (see below). $H$ is the Lie group of the principal bundle associated with the Carrollian geometry (one for each 1-form $W$) that we will discuss in the paper. 

The Lie algebra  $\mathfrak{g}=\mathfrak{carr}(1+3)
=\mathfrak{so}(3)_J\ltimes
\bigl(\mathbb{R}^3_C\ltimes(\mathbb{R}_E\times\mathbb{R}^3_P)\bigr)$, where the subscripts refer to the standard generators, is given by
\begin{equation}\label{eq: algebra}
 \begin{array}{llll}
     [E,E]=0\,, \quad& [E,P_j]=0\,,\quad&[E,C_j]=0\,, &[E,J_j]=0\,,\\[.3ex]
      &[P_i,P_j]=0\,,&[P_i,C_j]=\delta_{ij}E\,,\quad&[P_i,J_j]=\varepsilon_{ijk}P^k\,,\\[.3ex]
      &&[C_i,C_j]=0\,,&[C_i,J_j]=\varepsilon_{ijk}C^k\,,\\[.3ex]
      &&&[J_i,J_j]=\varepsilon_{ijk}J^k\,.
 \end{array} 
\end{equation}
In practice, one can think of $X=X_J+X_C+X_E+X_P\in \mathfrak{carr}(1+3)$ as a matrix: 
\[
X=\left(\begin{array}{ccc}
   0 & 0^t & 0 \\
   X_P^ib_i  & X_J^iJ_i & 0\\
   X_E & -X_C^ib_i^t & 0
\end{array}\right)
\]
where 
\[
P_i=\left(\begin{array}{ccc}
   0 & 0 & 0 \\
 b_i  & 0 & 0\\
   0 & 0 & 0
\end{array}\right)\,,\quad 
C_i=\left(\begin{array}{ccc}
   0 & 0 & 0 \\
   0  &0 & 0\\
   0 & -b_i^t & 0
\end{array}\right)\,,\quad J_i\cong \left(\begin{array}{ccc}
   0 & 0 & 0 \\
   0  &J_i & 0\\
   0 & 0 & 0
\end{array}\right)\,,\quad E=\left(\begin{array}{ccc}
   0 & 0 & 0 \\
   0  &0 & 0\\
   1 & 0 & 0
\end{array}\right)
\]
and $b_i$ is the $i$-th element of the canonical basis of $\mathbb{R}^3$.  The subspaces 
\(
\mathfrak h=\mathrm{span}(J_i,C_i)\) and \(
\mathfrak m=\mathrm{span}(E,P_i),
\)
define the reductive decomposition
\[
\mathfrak{g}=\mathfrak{h}\oplus\mathfrak{m}\,,
\qquad\text{with}\qquad
\begin{aligned}
  [\mathfrak{h},\mathfrak{h}]\subset\mathfrak{h}\,, \qquad & [\mathfrak{h},\mathfrak{m}]\subset\mathfrak{m}\,, \\
  & [\mathfrak{m},\mathfrak{m}]=0\,.
\end{aligned}
\]
Using this reductive decomposition and the standard bases in each sector, we write
\begin{equation}\label{eq: W=Wh+Wm}
W=W_{\mathfrak h}+W_{\mathfrak m},
\qquad\qquad\begin{array}{l}
 W_{\mathfrak h}=\Aa+\Upomega
=\Aa^iJ_i+\Upomega^iC_i, \\ 
W_{\mathfrak m}=\ee+\uptau
=\ee^iP_i+\uptau E.
\end{array} 
\end{equation}
The Cartan curvature (see Appendix~\ref{sec_math_results} for the notation)
\[
F_W=\ddd W+\frac12[W\wed W]
\]
decomposes, projecting over $\mathfrak{h}$ and $\mathfrak{m}$,  as
\begin{equation}\label{eq: F_W split}
F_W=(F_W)_{\mathfrak h}+(F_W)_{\mathfrak m}
=F_{W_{\mathfrak{h}}}+\ddd_{W_{\mathfrak{h}}}W_{\mathfrak{m}}=(F_\Aa+\ddd_\Aa\Upomega)
+\bigl(\ddd_\Aa \ee+\ddd\uptau+[\Upomega\wed \ee]\bigr)
\end{equation}
The term in the second parentheses corresponds to the torsion 2-form of the solder form $W_{\mathfrak{m}}=\uptau+\ee$:
\begin{equation}
    T=\ddd_{W_{\mathfrak{h}}} W_{\mathfrak{m}} = \ddd_\Aa \ee + \ddd \uptau + [\Upomega \wed \ee]\label{torsion_W}
\end{equation}
Notice that, if we require the torsion to vanish, we have the condition $\ddd_{W_{\mathfrak{h}}} W_{\mathfrak{m}} = 0$, which as an equation for \( W_{\mathfrak{h}} \) might be incompatible. Moreover, even if a solution exists, it may not uniquely determine \( \Upomega \).

The 1-parameter gauge transformations 
\(W\mapsto \mathrm{Ad}(h_t^{-1})\circ W+h_t^{-1}\ddd h_t\) associated with $h_t=\exp(t X_\mathfrak{h})$, where $X_\mathfrak{h}\in \Omega^0(M,\mathfrak{h})$, define the infinitesimal gauge transformation  [notice that $\mathfrak{gau}(M\times H)=\Omega^0(M,\mathfrak{h})$]:
\[
\delta_{(X_{\mathfrak{h}})}W=\ddd_{W}X_{\mathfrak{h}}=\ddd X_{\mathfrak{h}} +[W,X_{\mathfrak{h}}],
\]
whose components are
\begin{align*}
   \delta_{(X_{\mathfrak{h}})} W_{\mathfrak{h}} :=(\delta_{(X_\mathfrak{h})} W)_{\mathfrak{h}}&=\ddd_{W_\mathfrak{h}} X_\mathfrak{h},\quad 
\delta_{(X_\mathfrak{h})} W_{\mathfrak{m}}:=(\delta_{(X_\mathfrak{h})} W)_{\mathfrak{m}}=-[X_{\mathfrak{h}}, W_{\mathfrak{m}}].
\end{align*}
These are the standard gauge transformations on the $H$-principal bundle, that must be gauge symmetries of our models. Typically, one considers the (non-$\mathrm{Ad}(H)$-invariant) gauge transformations of the coordinates $\{\Aa,\Upomega,\ee,\uptau\}$ associated with the basis $\{J,C,P,E\}$. Splitting the gauge parameter $X_{\mathfrak{h}}=X_J+X_C$ and using \eqref{eq: algebra}:
\begin{equation}\label{gauge_W}
\begin{aligned}
 &\delta_{(X_{\mathfrak{h}})}\Aa:= (\delta_{(X_{\mathfrak{h}})} W_{\mathfrak{h}})_J=\ddd_\Aa X_J\,,&&\delta_{(X_{\mathfrak{h}})}\ee := (\delta_{(X_{\mathfrak{h}})} W_{\mathfrak{m}})_{P}= -[X_J,\ee]\,,\\    &\delta_{(X_{\mathfrak{h}})} \Upomega:=(\delta_{(X_{\mathfrak{h}})} W_{\mathfrak{h}})_{C}= \ddd_\Aa X_C -[X_J,\Upomega]\,,&&
 \delta_{(X_{\mathfrak{h}})}\uptau :=(\delta_{(X_\mathfrak{h})} W_{\mathfrak{m}})_{E}=-[X_C,\ee]\,.
\end{aligned}
\end{equation}
The transformation \(
\delta_{(X_{\mathfrak{h}})}W=\ddd_{W}X_{\mathfrak{h}}\)  can be \textit{generalized} to \(\delta_{(X)}W=\ddd_{W}X\) for an arbitrary gauge parameter
\(X=X_{\mathfrak m}+X_{\mathfrak h}\),
where
\(X_{\mathfrak m}\)
and
\(X_{\mathfrak h}\)
denote the translational and vertical parts, respectively. This induces the transformations 
\begin{equation}\label{eq: generalized delta}
  \delta_{(X)} W_{\mathfrak{h}}  :=\ddd_{W_{\mathfrak{h}}} X_{\mathfrak{h}}\,,\qquad  \delta_{(X)} W_{\mathfrak{m}}  :=\ddd_{W_{\mathfrak{h}}} X_{\mathfrak{m}} -[X_{\mathfrak{h}},W_{\mathfrak{m}}]\,.
\end{equation}
The generalized gauge transformation admits a simple geometric interpretation. On the one hand, using Cartan's magic formula, the Lie derivative can be rewritten as
\begin{equation}\label{eq: L_xi Cartan}
    \pounds_\upxi W= \ddd_W(\upxi\iprod W)+\upxi\iprod F_W\,.
\end{equation}
On the other hand, 
 for any vector field $\upxi\in \mathfrak{X}(M)$, the Lie derivative of $W$ splits as
\begin{align*}
   \pounds_{\upxi}W&= ( \pounds_{\upxi}W)_{\mathfrak{h}}+ ( \pounds_{\upxi}W)_{\mathfrak{m}}= \pounds_{\upxi}W_{\mathfrak{h}}+  \pounds_{\upxi}W_{\mathfrak{m}}.
\end{align*}
Therefore, using \eqref{eq: F_W split} and \eqref{eq: L_xi Cartan}, these components are given by 
\begin{align*}
    \pounds_{\upxi} W_\mathfrak{h}
    &= \ddd_{W_\mathfrak{h}} (\upxi \iprod W_{\mathfrak{h}}) +\upxi \iprod \left(F_{W_\mathfrak{h}}\right) \overset{\eqref{eq: generalized delta}}{=} \delta_{(\upxi\iprod W)} W_{\mathfrak{h}} +\upxi \iprod \left(F_{W_\mathfrak{h}}\right)\,,\\
     \pounds_{\upxi} W_\mathfrak{m}
     &= \ddd_{W_\mathfrak{h}} (\upxi \iprod W_{\mathfrak{m}})-[(\upxi\iprod W_\mathfrak{h})\wed W_\mathfrak{m}]+\upxi\iprod \left(\ddd_{W_\mathfrak{h}}W_\mathfrak{m}\right) \overset{\eqref{eq: generalized delta}}{=}  \delta_{(\upxi\iprod W)} W_{\mathfrak{m}}+\upxi\iprod \left(\ddd_{ W_\mathfrak{h}} W_\mathfrak{m}\right).
\end{align*}

When the Cartan geometry is torsion-free, \(
\ddd_{W_\mathfrak h}W_\mathfrak m=0
\), and these expressions read
\begin{align*}
    \pounds_{\upxi} W_\mathfrak{h}
    &= \delta_{(\upxi\iprod W)} W_{\mathfrak{h}} +\upxi \iprod \left(F_{W_\mathfrak{h}}\right)\,,\\
     \pounds_{\upxi} W_\mathfrak{m}
     &=  \delta_{(\upxi\iprod W)} W_{\mathfrak{m}}\,.
\end{align*}
The previous reasoning relates generalized gauge transformations to infinitesimal diffeomorphisms of the base manifold. Explicitly, by writing a generalized gauge 0-form as $X=X_J+X_C+X_P+X_E$, we have the analog of \eqref{gauge_W} for generalized transformations:
\begin{align*}
 &\delta_{(X)} \Aa= \ddd_\Aa X_J\,,&&\delta_{(X)} \ee= \ddd_\Aa X_P -[X_J,\ee]\,,\\   
 &\delta_{(X)} \Upomega=\ddd_\Aa X_C -[X_J,\Upomega]\,,&&
 \delta_{(X)} \uptau = \ddd X_E -[X_P,\Upomega]- [X_C,\ee]\,.
\end{align*}
Therefore, identifying $X_J=\upxi\iprod \Aa$, $X_C=\upxi\iprod \Upomega$, $X_P=\upxi\iprod \ee$, and $X_E=\upxi\iprod \uptau$, the action of an infinitesimal diffeomorphism generated by $\upxi$ can be expressed as
\begin{subequations}\label{diff_W}
\begin{align}
 \begin{split}
   \pounds_\upxi \Aa &=  \ddd_\Aa(\upxi\iprod \Aa) +\upxi\iprod F_\Aa \\
   &= \delta_{(\upxi\iprod \Aa)} \Aa +\upxi\iprod F_\Aa \,,
 \end{split} \label{diff_W_A} \\[1ex]
 \begin{split}
   \pounds_\upxi \Upomega &=  \ddd_\Aa(\upxi\iprod \Upomega)-[(\upxi\iprod \Aa),\Upomega]+\upxi\iprod( \ddd_\Aa\Upomega) \\
   &= \delta_{(\upxi\iprod \Aa+\upxi\iprod\Upomega)} \Upomega +\upxi\iprod ( \ddd_\Aa\Upomega)\,,
 \end{split} \label{diff_W_Omega} \\[1ex]
 \begin{split}
   \pounds_\upxi \ee &=  \ddd_\Aa(\upxi\iprod \ee)-[(\upxi\iprod \Aa),\ee]+\upxi\iprod ( \ddd_\Aa\ee) \\
   &= \delta_{(\upxi\iprod \Aa +\upxi\iprod \ee)} \ee +\upxi\iprod ( \ddd_\Aa\ee)\,,
 \end{split} \label{diff_W_e} \\[1ex]
 \begin{split}
   \pounds_\upxi\uptau &= \ddd(\upxi\iprod \uptau)- [(\upxi\iprod \Upomega),\ee]-[(\upxi\iprod \ee),\Upomega] +\upxi\iprod(  \ddd\uptau +[\Upomega\wed \ee]) \\
   &= \delta_{(\upxi\iprod\Upomega+\upxi\iprod \ee+\upxi\iprod\uptau)}\uptau +\upxi \iprod ( \ddd\uptau +[\Upomega\wed \ee])\,.
 \end{split} \label{diff_W_tau}
\end{align}
\end{subequations}
For torsion-free geometries, the last two expressions get simplified:
\[
\pounds_\upxi \ee= \delta_{(\upxi\iprod \Aa +\upxi\iprod \ee)} \ee\,,\qquad \pounds_\upxi\uptau=\delta_{(\upxi\iprod\Upomega+\upxi\iprod \ee+\upxi\iprod\uptau)}\uptau\,.
\]
In this paper, we work with diff-invariant actions, therefore these expressions will show up in the analysis of symmetries within the Hamiltonian set up.

%
%
	
\section{The action}\label{sec_Action}
Let us consider $M=[t_1,t_2]\times\Sigma$ a 4-dimensional manifold foliated by ``spatial slices'' $\Sigma_t:=\{t\}\times\Sigma$ ($t\in[t_1,t_2]$), where $\Sigma$ is a closed, orientable, and 3-dimensional manifold (as a consequence, parallelizable). On $M$, we have a canonical evolution vector field $\partial_{\mathrm{t}}\in \mathfrak{X}(M)$, transverse to every $\Sigma_t$ and  defined by the tangent vectors to the curves $c_p:\mathbb{R}\rightarrow M=[t_1,t_2]\times \Sigma:t\mapsto (t,p)$; and a scalar field $t\in C^\infty(M)$ defined as $t:M\rightarrow \mathbb{R}:(\tau,p)\mapsto \tau$, such that ${\partial_{\mathrm{t}}}\iprod \bd t=1$. For each $t\in[t_1,t_2]$, we define the embedding $\jmath_t:\Sigma\rightarrow M:p\mapsto(t,p)$. Its pullback will be denoted as $\jmath_t^\ast$. Notice that we have $\Sigma_t=\jmath_t(\Sigma)$. 

\noindent Let us consider the action
\begin{equation}\label{action_Carroll}
S(\Aa,\Upomega,\ee,\uptau)=\frac{1}{2}\int_{M}\Big( \kappa^{-1}\big(2\uptau\wedge \langle \ee\wed F_\Aa \rangle-\langle [\ee\wed\ee]\wed\ddd_A\Upomega  \rangle \big)+\gamma^{-1}\langle [\ee\wed\ee]\wed F_\Aa \rangle \Big)\,.
\end{equation}
The basic fields are $\Aa, \Upomega,  \ee  \in\Omega^1(M,\mathfrak{su}(2))$ and $\uptau  \in\Omega^1(M)$, with \(F_{\Aa}:=\bd\Aa+\frac{1}{2}[\Aa\wed\Aa]\) and \(\ddd_\Aa\Upomega:=\bd\Upomega+[\Aa\wed\Upomega]\,.\) Here $\mathrm{d}$ denotes the exterior differential in $M$ and $\kappa$, $\gamma\neq0$ are real constants. Notice that $\gamma$ can be interpreted as the Immirzi parameter since it comes from the Holst term. 

Before we go ahead with the study of this action, some comments are in order. Although the original fields are $\mathfrak{carr}$-valued differential forms, when the basis defined in Section \ref{sec_geometric} is used to write the limit of the Holst action \cite{BarberoG:2025rxm}, leaving $\uptau$ aside, the resulting final expression is written in terms of $\mathbb{R}^3$-valued differential forms $\ee^i$, $\Upomega^i$, $\Aa^i$ (i.e. carrying a single “internal index”). They  appear combined in expressions of the form \(\ee^i\wedge  \Upomega^j \delta_{ij}\), \(\epsilon_{ijk} \ee^i \wedge \ee^j\wedge \Aa^k\), etc. that can be interpreted as $\langle \ee \wed \Upomega \rangle$, $[\ee\wed \ee]$, etc. In practice, these expressions are exactly the same as those obtained by working with fields \(\ee = \ee^i J_i\), \(\Upomega = \Upomega^i J_i\), and \(\Aa = \Aa^i J_i\) taking values in \(\mathfrak{su}(2)\), and by considering the operations used to write the action \eqref{action_Carroll}. However, it is possible to map the corresponding parts of the algebra onto $\mathfrak{su}(2)$ and effectively work as if all the fields used to write the action are $\mathfrak{su}(2)$-valued differential forms. Second, we want to note an interesting fact: the action \eqref{action_Carroll} has the form of the (anti)self-dual action for Euclidean GR with an extra term. Specifically, if we choose the coupling constants $\kappa$ and $\gamma$ such that $\kappa=-\gamma$ the first and third terms reproduce the self-dual action in the form appearing in \cite{BarberoG:1994fcp, BarberoG:2023qih}.

In addition to being invariant under 4-dimensional diffeomorphisms, the action \eqref{action_Carroll} is invariant under two classes of local internal transformations. The first corresponds to the infinitesimal $\mathfrak{su}(2)$ gauge rotations introduced in Section Section \ref{sec_geometric} [cf. Eq.~\eqref{gauge_W}], which act on the fields according to
\begin{equation}\label{internal_rotations}
\begin{aligned}
    &\delta^1_{(Z)}\Aa := \mathrm{d}_\Aa Z \,,\qquad\qquad&&\delta^1_{(Z)}\ee := -[Z, \ee] \,,\\
    &\delta^1_{(Z)}\Upomega := -[Z, \Upomega] \,,&&    \delta^1_{(Z)}\uptau := 0 \,,
\end{aligned}
\end{equation}
where $Z \in \Omega^0(\Sigma, \mathfrak{su}(2))$. The second class comprises the local Carrollian boosts parameterized by $y \in \Omega^0(\Sigma, \mathfrak{su}(2))$:
\begin{equation}\label{Carrollian_boosts}
\begin{aligned}
    &\delta^2_{(y)}\Aa:= 0 \,,&&\delta^2_{(y)}\ee := 0 \,,\\
    &\delta^2_{(y)}\Upomega := \mathrm{d}_\Aa y \,,\qquad\qquad&&    \delta^2_{(y)}\uptau := - \langle y, \ee \rangle \,.
\end{aligned}
\end{equation}
From a geometric perspective, the gauge transformations $\delta^1_{(Z)}$ and $\delta^2_{(y)}$ originate from the fundamental rotations ($J$) and Carrollian boosts ($C$) of the underlying kinematic algebra before performing the reduction to its $\mathfrak{su}(2)$ sector. 

As explicitly shown in Appendix \ref{app_juan4}, these internal symmetries are deeply intertwined with spacetime geometry. When performing a $1+3$ canonical splitting of the 4-dimensional manifold, the full spacetime diffeomorphisms systematically project onto the spatial hypersurface $\Sigma$ as a combination of spatial diffeos and these internal gauge and boost transformations, with the effective transformation parameters matching the fields' projections defined later in \eqref{shorthand}.                                   

The dynamics of the theory is encoded in the following set of field equations:
\begin{subequations}\label{field_eqs}
\begin{align}
  & \kappa^{-1}\ddd_\Aa[\ee\wed \ee]=0\,, \\
  & 2\kappa^{-1}\ddd_\Aa(\uptau\wedge \ee)-\kappa^{-1}[\Upomega\wed[\ee\wed\ee]]+\gamma^{-1} \ddd_\Aa[\ee\wed\ee]=0\,, \\
  & \gamma^{-1}[\ee\wed F_\Aa]-\kappa^{-1} F_\Aa\wedge\uptau-\kappa^{-1} [\ee\wed \ddd_\Aa\Upomega ]=0\,, \\
  & \kappa^{-1}\langle\ee\wed F_\Aa\rangle=0\,.
\end{align}
\end{subequations}

By clearing the constant coupling constants $\kappa$ and $\gamma$ these expressions can be simplified to:
\begin{subequations}\label{field_eqs_2}
\begin{align}
  & \ddd_\Aa[\ee\wed \ee]=0\,, \\
  & 2\ddd_\Aa(\uptau\wedge \ee)=[\Upomega\wed[\ee\wed\ee]]\,, \\
  & [\ee\wed(\kappa F_\Aa-\gamma \ddd_\Aa\Upomega )]=\gamma F_\Aa\wedge\uptau\,, \\
  & \langle\ee\wed F_\Aa\rangle=0\,.
\end{align}
\end{subequations}

The first equation imposes a covariant constancy condition on the spatial soldering forms. The second equation relates the covariant derivative of the mixed $\uptau\wedge\ee$ term to a quadratic expression in $\Upomega$ and $\ee$, effectively constraining the torsion. The third equation couples the curvature $F_\Aa$ with the covariant derivative of $\Upomega$, mixing the isotropy and translational sectors. Finally, the fourth equation provides an additional algebraic constraint on the curvature.

If in \eqref{field_eqs} we take the limit $\kappa\rightarrow\infty$, we recover the field equations for the HK model \cite{Husain:1990vz}:
\begin{subequations}
\begin{align}\label{field_eqs_HK}
  & [\ee\wed F_\Aa]=0\,, \\
  & \ddd_\Aa[\ee\wed\ee]=0\,.
\end{align}
\end{subequations}
These equations describe a theory related to GR with simple dynamics and a Hamiltonian formulation in terms of Ashtekar variables.

If, instead, we take the limit $\gamma\rightarrow\infty$, the field equations become
\begin{subequations}\begin{align}\label{field_eqs_EP}
  & \ddd_\Aa[\ee\wed \ee]=0\,, \\
  & 2\ddd_\Aa(\uptau\wed \ee)-[\Upomega\wed[\ee\wed\ee]]=0\,, \\
  & F_\Aa\wedge \uptau+[\ee\wed \ddd_\Aa\Upomega ]=0\,, \\
  & \langle\ee\wed F_\Aa\rangle=0\,.
\end{align}
\end{subequations}
These are the equations given by the Einstein-Cartan action.

It is worth noting that the field equations for Euclidean GR in terms of the fields used here can be written in the form
\begin{subequations}
\begin{align}\label{field_eqs_EuclGR}
  & [\ee\wed F_\Aa]=-F_\Aa\wedge\uptau\,, \\
  & 2\ddd_\Aa(\uptau\wedge\ee)=\ddd_\Aa[\ee\wed \ee]\,, \\
  & \langle\ee\wed F_\Aa\rangle=0\,.
\end{align}
\end{subequations}
where the terms that appear in the field equations coming from the action \eqref{action_Carroll} are mixed in interesting ways.

In order to facilitate the comparison of our results with the literature on the Hamiltonian treatment of gravitational actions similar to the one discussed here, we will take $\kappa=1$ but keep the Immirzi parameter $\gamma$ in the rest of the paper.
As shown in \ref{subsec_simplif}, the equations \eqref{field_eqs_2} can be simplified to 
\begin{subequations}\begin{align}
  & \ddd_\Aa\ee=0\,, \label{EOM_1}\\
  & \ddd\uptau+\langle\Upomega\wed \ee\rangle=0\,, \label{EOM_2}\\
  & F_\Aa\wedge \uptau+[\ee\wed\ddd_\Aa\Upomega]=0\,,\label{EOM_3}\\
  & \langle\ee\wed F_\Aa\rangle=0\,.\label{EOM_4}
\end{align}
\end{subequations}
Notice that the dependence on $\gamma$ drops out. This phenomenon is similar to the one that happens with the Holst action. As in that case, the term of the action controlled by the Immirzi parameter does not change the field equations. 

The system that we study in the paper admits an interesting geometric interpretation. From equation \eqref{torsion_W}, it is clear that equations \eqref{EOM_1} and \eqref{EOM_2} are equivalent to the vanishing of the torsion (a condition that ties together the translational and isotropy sectors). In particular, the first equation imposes the covariant constancy of the spatial coframe $\ee$, while the second relates the covariant derivative of the clock form $\uptau$ to a contraction involving $\Upomega$ and $\ee$  \cite{figueroa2020intrinsic}. The remaining equations \eqref{EOM_3} and \eqref{EOM_4} govern the curvature dynamics.

Although the main purpose of the paper is to study the Hamiltonian formulation for the action \eqref{action_Carroll}, there is interesting information that can be gleaned from the field equations. By introducing a foliation of the spacetime manifold $M$ (or using the one available by construction) and defining appropriate projected geometric objects in terms of $\{\Aa,\Upomega,\ee,\uptau\}$, it is possible to get a rather clear idea about how the Hamiltonian formulation will look like, both as far as the constraints and the equations of motion are concerned. In the rest of the section we will look at a specific feature of the dynamics given by \eqref{field_eqs_2} connected to a concrete description of the collapsed Carrollian light cones.

Let us restrict ourselves to solutions to the field equations such that $(\ee,\uptau)$ defines a (non-degenerate) coframe. We will then have a dual frame $(\mathsf{u}, \mathsf{E}_i)$ such that $\mathsf{u} \iprod\uptau=1$, $\mathsf{u}\iprod\ee=0$, $\mathsf{E}_i\iprod\uptau=0$ and $\mathsf{E}_i\iprod\ee^j=\delta^j_i$. From Lemma \ref{lemma_[e,B]=0} and $[F_\Aa\wed\ee]=0$ (which follows from $\ddd_A\ee=0$), we have $\uu\iprod\ddd_\Aa\ee=0$ and $\uu\iprod F_\Aa=0$. We can then use \eqref{diff_W} to obtain the Lie derivatives along $\mathsf{u}$ of solutions $\{\Aa,\Upomega,\ee,\uptau\}$ to the field equations \eqref{field_eqs_2}:
\begin{align*}
\pounds_{\uu}\Aa&=\ddd_{\Aa}(\uu\iprod\Aa) \,,\\
  \pounds_{\uu}\Upomega &=\ddd_{\Aa} (\uu\iprod\Upomega)-[(\uu\iprod\Aa)\wed\Upomega]-(*_\ee F_{\Aa})^\top\,,
  \\
  \pounds_{\uu}\ee\, &=-[(\uu\iprod\Aa)\wed \ee)] \,,\\
  \pounds_{\uu}\uptau\,&=- \langle(\uu\iprod\Upomega)\wed\ee\rangle \,.
\end{align*}
In the last equation, we have used $\pounds_{\uu}\uptau=\uu\iprod \ddd \uptau =-\uu\iprod \langle\Upomega\wed \ee\rangle $ whereas, in the second one, we have used $[\ee\wed(\mathsf{u}\iprod\ddd_{\Aa}\Upomega)]=F_\Aa$ and Lemmas  \ref{lemma_[e,B]=0} and \ref{lemma_[e,X]=C}, which makes use of the following notation: given $\mathsf{B}=\frac{1}{2} \tensor{\mathsf{B}}{_i_j^k} \ee^i\wedge \ee^j J_k$, we denote 
\begin{equation}\label{eq: star B and B^T}
 \star_\ee \mathsf{B}:=\frac{1}{2}\tensor{\epsilon}{_i^k^l}\tensor{\mathsf{B}}{_k _l ^j} (\ee^i J_j)\,,\qquad 
 (\star_\ee \mathsf{B})^\top:= \frac{1}{2}\tensor{\epsilon}{^i^k^l}\tensor{\mathsf{B}}{_k _l _j} (\ee^j J_i)
\end{equation} 
where $\star_{\ee}$ is the Hodge star operator defined with respect to the spatial coframe basis $\ee^i$, while, for a given $\mathsf{X}=\tensor{\mathsf{X}}{_i^j}\ee^iJ_j$, we define its ``transpose'' as $\mathsf{X}^\top=\tensor{\mathsf{X}}{^j_i}\ee^iJ_j$. 
 
As we can see, these expressions give a very neat picture of the dynamics of magnetic Carrollian gravity \cite{Campoleoni:2022ebj}. With the exception of $\pounds_{\uu}\Upomega$, where a contribution $-(*_{\ee} F_{\Aa})^\top$ is present, the rest of the Lie derivatives along solutions to the field equations can be interpreted as either internal $SU(2)$ rotations with parameter $(\uu\iprod\Aa)$ or Carrollian boosts with parameter $(\uu\iprod\Upomega)$. In view of this result, the integral curves of $\uu$ can be interpreted as the collapsed light cones.

We end this section by noting that the field equations do not imply the integrability of the distribution
\[
\mathcal{T}=\{v\in TM\,:\, \uptau(v)=0\}.
\]
Indeed, if we decompose
\(
\Upomega=\uptau\wedge (\mathrm{u}\iprod \Upomega)+\underline{\Upomega},
\)
and demand the integrability of $\mathcal{T}$:
\[
0=\uptau\wedge\mathrm{d}\uptau
\overset{\eqref{EOM_2}}{=}
-\uptau\wedge\langle \Upomega\wed \ee \rangle,
\]
we get the additional condition 
\[
\langle\underline{\Upomega}\wed \ee\rangle=0.
\]
We will recover an analog of this in Section~\ref{sec_time_gauge}, when we discuss the time gauge in the Hamiltonian formulation.

%
%
\section{The Hamiltonian formulation}\label{sec_Hamiltonian}

We now obtain the Hamiltonian formulation starting from the Lagrangian. To this end we take advantage of the fact that the spacetime manifold $M$ has the form $[t_1,t_2]\times\Sigma$ and use the scalar field $t\in C^\infty(M)$ and the vector field $\partial_{\mathrm{t}}$ canonically associated with $M$.  The action can be written as the integral of a 4-form $\mathcal{L}$ over $M$
\[
\int_{[t_1,t_2]\times\Sigma}\mathcal{L}=\int_{t_1}^{t_2}\mathrm{d}t\int_\Sigma \jmath_t^*\partial_{\mathrm{t}}\iprod \mathcal{L}\,.
\]
From this expression, we can read the Lagrangian $L:=\int_\Sigma\jmath_t^*\partial_{\mathrm{t}}\iprod \mathcal{L}$. This is a real function on the tangent bundle of the configuration space [here $\mathfrak{g}=\mathfrak{su}(2)$]
\[
Q\subset V
:=(\Omega^0(\Sigma,\mathfrak{g})\times \Omega^1(\Sigma,\mathfrak{g}))^3\times(\Omega^0(\Sigma)\times\Omega^1(\Sigma))\,.
\]
given by the elements of the form $q:=(e_{\mathrm{t}},e,A_{\mathrm{t}},A,\Omega_{\mathrm{t}},\Omega,\tau_{\mathrm{t}},\tau)$ subject to the non-degeneracy condition for $\{e_{\mathrm{t}},e,\tau_{\mathrm{t}},\tau\}$ discussed before [see equations \eqref{eq: W_m iso} and \eqref{eq: W=Wh+Wm}]. Since these restrictions can be expressed as inequalities, they define an open (and in fact dense) subset, thereby preserving the local dimensionality. The velocities have the form $\mathrm{v}_q:=(v_{e_{\mathrm{t}}},v_e,v_{A_{\mathrm{t}}},v_{A},v_{\Omega_{\mathrm{t}}},v_{\Omega},v_{\tau_{\mathrm{t}}},v_{\tau})$.

By computing $\jmath_t^*\partial_{\mathrm{t}}\iprod \mathcal{L}$ from the action \eqref{action_Carroll} we obtain the Lagrangian
\begin{align}\label{Lagrangian_Carroll}
\begin{split}
L(\mathrm{v}_q)&=\frac{1}{2}\int_\Sigma \Big(\,\,\langle v_A\wed(\gamma^{-1}[e\wed e]-2 e\wedge \tau)\rangle-\langle v_\Omega\wed[e\wed e]\rangle\\
& \hspace*{1.3cm}+\langle 2e_{\mathrm{t}},(\gamma^{-1}[e\wed F_A]-[e\wed\mathrm{d}_A\Omega]-(F_A\wedge \tau)) \rangle\phantom{\Big{|}}\\
&\hspace*{1.3cm}+\langle A_{\mathrm{t}}\wed(\gamma^{-1}\mathrm{d}_A[e\wed e]-[\Omega\wed [e\wed e]]-2\mathrm{d}_A(e\wedge \tau)) \rangle\\
&\hspace*{1.3cm}-\langle\Omega_{\mathrm{t}}\wed\mathrm{d}_A[e\wed e]\rangle+2\tau_{\mathrm{t}}\wedge\langle e\wed F_A\rangle\Big)\,.
\end{split}
\end{align}
Here we have introduced the curvature
$F_A:=\mathrm{d}A+\frac{1}{2}[A\wed A]$, and the covariant exterior differential $\mathrm{d}_A$ which acts on $\alpha\in\Omega^k(\Sigma,\mathfrak{g})$ as $
\mathrm{d}_A\alpha:=\mathrm{d}\alpha+[A\wed \alpha]\,.$

The fiber derivative of the Lagrangian $FL:TQ\rightarrow T^*Q$ is given by
\[
FL(\mathrm{v}_q)(\mathrm{w}_q)=\frac{1}{2}\int_\Sigma\left(\langle w_A\wed(\gamma^{-1}[e\wed e]-2 e\wedge \tau)\rangle-\langle w_\Omega\wed[e\wed e]\rangle\right)\,.
\]
This provides the definition of momenta. The non-zero components are 
\begin{align}\label{momenta}
\begin{split}
\,\,{\mathrm{p}}_{A}&=\frac{1}{2} (\gamma^{-1}[e\wed e]-2(e\wedge\tau))\\
{\mathrm{p}}_\Omega&=-\frac{1}{2} [e\wed e]
\end{split}
\end{align}
where we consider the obvious pairing to define their action on vectors. Since no velocities are involved, all the momenta can be removed and the primary constraint submanifold $\mathsf{M}_0$ is diffeomorphic to $Q$. In particular, the components of the momenta of a generic vector field $\mathbb{Y}$ in phase space can be rewritten in terms of the components of the configuration variables. 

The non-vanishing components of the vector fields tangent to the primary constraint submanifold are 
\begin{align}\label{tangent_vect}
\begin{split}
&{\mathrm{Y}}_{\!\mathrm{p}A}=\gamma^{-1}[Y_e\wed e]-\, Y_e\wedge \tau-\, e\wedge Y_\tau\,, \\
&{\mathrm{Y}}_{\!\mathrm{p}\Omega}=-[Y_e\wed e]\,,
\end{split}
\end{align}
Therefore, the canonical symplectic form pulled back to  $\mathsf{M}_0$  reads
\begin{align}\label{Omega_tan}
\begin{split}
\omega_0(\mathbb{X}_0,\mathbb{Y}_0)=&\int_\Sigma\Big(\langle Y_e\wed\big( \gamma^{-1}[e\wed X_A]-\,\tau\wedge X_A-[e\wed X_\Omega]  \big)  \rangle\\
&\hspace*{4mm}+\langle Y_A\wed\big( X_e\wedge\tau+\, e\wedge X_\tau-\gamma^{-1}[X_e\wed e]\big) \rangle\\
&\hspace*{4mm}+\langle Y_\Omega\wed[X_e\wed e]\rangle\\
&\hspace*{4mm}+Y_\tau\wedge\langle e\wed X_A\rangle\Big)\,.
\end{split}
\end{align}
for a generic field  $\mathbb{X}_0=(X_{e_{\mathrm{t}}}, \!X_e,X_{A_{\mathrm{t}}}, \!X_A,\!X_{\Omega_{\mathrm{t}}}, \!X_\Omega,\!X_{\tau_{\mathrm{t}}},\!X_\tau)$ and analogously for $\mathbb{Y}_0$.

Meanwhile, since the energy $E(v)= FL(v)(v)-L(v)$ only depends on the positions, it is equal to  the Hamiltonian which is given by
\begin{align}\label{Hamiltonian_Carroll}
\begin{split}
H(q)&=\frac{1}{2}\int_\Sigma \Big(\langle 2e_{\mathrm{t}},([e\wed\mathrm{d}_A\Omega]+(F_A\wedge \tau)-\gamma^{-1}[e\wed F_A]) \rangle\phantom{\Big{|}}\\
&\hspace*{1.2cm}+\langle A_{\mathrm{t}}\wed([\Omega\wed [e\wed e]]+2\mathrm{d}_A(e\wedge \tau)-\gamma^{-1}\mathrm{d}_A[e\wed e]) \rangle\\
&\hspace*{1.2cm}+\langle\Omega_{\mathrm{t}}\wed\mathrm{d}_A[e\wed e]\rangle-2\tau_{\mathrm{t}}\wedge\langle e\wed F_A\rangle\Big)\,.
\end{split}
\end{align}

The main equation that has to be solved in the GNH setting is $(\mathbb{X}\iiiprod\Omega-\dd H)|{\mathsf{M}_0}=0$  (here $\dd$ and $\iprod\!\!\!\!\iprod$ respectively denote the exterior and interior differentials in phase space). To this end, we compute $\dd H(\mathbb{Y}_0)$ for an arbitrary vector field $\mathbb{Y}_0$ tangent to $\mathsf{M}_0$ and compare it to \eqref{Omega_tan}. This leads to the following \textit{secondary constraints}
\begin{subequations}
\begin{align}
&[e\wed(\gamma\,\mathrm{d}_A\Omega- F_A)]+\gamma F_A\wedge \tau=0\,,\label{sec_const_2}\\
&\mathrm{d}_A(2\gamma\, e\wedge\tau-[e\wed e])+\gamma[\Omega\wed[e\wed e]]=0\,,\label{sec_const_1}\\
&\,\mathrm{d}_A[e\wed e]=0\,,\label{sec_const_3}\\
&\langle e\wed F_A\rangle=0\,,\label{sec_const_4}
\end{align}
\end{subequations}
We also obtain the equations for $(X_e,X_A,X_\Omega,X_\tau)$:
\begin{subequations}\label{eq: Hamiltonian X}
\begin{align}
    \begin{split}
        [e\wed X_A]\!-\!\gamma\tau\!\wedge\! X_A\!-\!\gamma[e\wed X_\Omega] = &\,[(\gamma\,\mathrm{d}_A\Omega- F_A)\wed e_{\mathrm{t}}]+\gamma\,\mathrm{d}_AA_{\mathrm{t}}\wedge\tau+[e\wed \mathrm{d}_AA_{\mathrm{t}}] \\
        &+\gamma[e\wed [A_{\mathrm{t}}\wed \Omega]]+\gamma[e\wed\mathrm{d}_A\Omega_{\mathrm{t}}]-\gamma\,\tau_{\mathrm{t}}F_A\,,
    \end{split}\label{eq_X_2}\\[.5ex]
    \begin{split}
        \gamma X_e\wedge\tau\!+\!\gamma e\wedge X_\tau\!-\![X_e\wed e] = &-\,\mathrm{d}_A[e_{\mathrm{t}}\wed e]+\gamma[(e\wedge\tau)\wed A_{\mathrm{t}}]+[[A_{\mathrm{t}}\wed e]\wed e] \\
        &\mbox{}\hspace{-8ex}-\gamma[[\Omega_{\mathrm{t}}\wed e]\wed e]\!+\!\gamma\,\mathrm{d}_A(e_{\mathrm{t}}\wedge\tau\!-\!\tau_{\mathrm{t}}\wedge e)\!+\!\gamma[\Omega\wed[e_{\mathrm{t}}\wed e]]\,,
    \end{split}\label{eq_X_1}\\[.5ex]
    [X_e\wed e] = &\,\,\mathrm{d}_A[e_{\mathrm{t}}\wed e]+[[e\wed A_{\mathrm{t}}]\wed e]\,,\label{eq_X_3}\\
    \langle e\wed X_A\rangle = &\,\langle F_A\wed e_{\mathrm{t}}\rangle+\langle e\wed \mathrm{d}_A A_{\mathrm{t}}\rangle\,.\label{eq_X_4}
\end{align}
\end{subequations}

They can be significantly simplified by replacing the unknowns $(X_e,X_A,X_\Omega,X_\tau)$ with $(Z_e,Z_A,Z_\Omega,Z_\tau)$ according to
\begin{align}\label{new_variables_1}
\begin{split}
X_{\!A}&=\pounds_\xi A+\mathrm{d}_A\Lambda+Z_A\,,\\
X_{\!\Omega}&=\pounds_\xi\Omega-[\Lambda\wed \Omega]+\mathrm{d}_A\psi+Z_\Omega\,,\\
X_e&=\pounds_\xi e-[\Lambda\wed e]+Z_e\,,\\
X_\tau&=\pounds_\xi \tau-\langle\psi\wed e\rangle+\mathrm{d}N+Z_\tau\,,
\end{split}
\end{align}
where we have used the shorthand notation
\begin{equation}\label{shorthand}
\begin{aligned}
\Lambda&:=A_{\mathrm{t}}-\xi\iprod A\,,\qquad&&N:=\tau_{\mathrm{t}}-\xi\iprod\tau\,,\\
\psi&:=\Omega_{\mathrm{t}}-\xi\iprod\Omega\,,\qquad&&\xi:=\langle e_{\mathrm{t}}\wed e\rangle\,.
\end{aligned}
\end{equation}
Geometrically, $(\Lambda,\psi)$ are the local gauge parameters that generate the internal gauge symmetries [see equations \eqref{internal_rotations} and \eqref{Carrollian_boosts}], while $(N,\xi)$ encapsulate the $4$-dimensional diffeomorphisms (they are related to the lapse-shift; we require $N\neq 0$). 

The redefinitions \eqref{new_variables_1} are dictated by the symmetries of the action and previous results on the Husain-Kucha\v{r} model \cite{BarberoG:2025rxm}. The equations for the Hamiltonian vector fields \eqref{eq: Hamiltonian X} in terms of new variables $(Z_e,Z_A,Z_\Omega,Z_\tau)$ are:
\begin{subequations}\label{eq: eq_Z}
\begin{align}
&[Z_A\wed e]-\gamma \tau\wedge Z_A-\gamma [e\wed Z_\Omega]=-\gamma N F_A\label{eq_Z_2}\,,\\
&Z_e\wedge\tau+e\wedge Z_\tau=- N \mathrm{d}_A e\label{eq_Z_1}\,,\\
&[Z_e\wed e]=0\label{eq_Z_3}\,,\\
&\langle e\wed Z_A\rangle=0\label{eq_Z_4}\,.
\end{align}
\end{subequations}
We sketch their resolution now.

\medskip

\noindent \textbf{Step 1:} Using Lemma \ref{lemma_[e,X]=C}, we can solve equation \eqref{eq_Z_3} and obtain $Z_e=0$.
\smallskip

\noindent \textbf{Step 2:} We solve now \eqref{eq_Z_1} which, on account of the previous result, simply reads $e\wedge Z_\tau=- N \mathrm{d}_A e$. This is an inhomogeneous equation which can only be solved if the following condition holds (see Appendix C of \cite{BarberoG:2021ekv})
\begin{equation}\label{consistency_cond}
e^{(i}\wedge \mathrm{d}_A e^{j)}=0\,,
\end{equation}
in which case the solution is (see \ref{subsec_topforms} for notation)
\[
Z_\tau=-\frac{1}{2}N\left(\frac{[\mathrm{d}_A e\wed e]^i}{\mathrm{vol}_e}\right)e_i\overset{\eqref{sec_const_3}}{=}0\,.
\]
This implies that the right hand side of \eqref{eq_Z_1} must vanish too. Hence, recalling that $N\neq0$ everywhere on $\Sigma$, we have the additional secondary constraints $\mathrm{d}_Ae=0$.
\smallskip

\noindent \textbf{Step 3:} The new constraints imply that \eqref{sec_const_3} and \eqref{consistency_cond} hold. Moreover, we can simplify \eqref{sec_const_1} as follows:
\begin{equation}\label{eq: 2e wedge torsion_1}
0=2e\wedge\mathrm{d}\tau+[\Omega\wed[e\wed e]]=2e\wedge(\mathrm{d}\tau+\langle\Omega\wed e\rangle)\,,
\end{equation}
where we have used the identity $[\Omega\wed[e\wed e]]=-2\langle\Omega\wed e\rangle\wedge e$. It is straightforward to see now that \eqref{eq: 2e wedge torsion_1} is equivalent to $\mathrm{d}\tau+\langle\Omega\wed e\rangle=0$.

Finally, noticing that $\mathrm{d}_Ae=0$ implies $[F_A\wed e]=0$, we can also simplify  \eqref{sec_const_2} leading to the secondary constraint submanifold $\mathsf{M}_1$ defined by
\begin{subequations}\label{sec_const}
\begin{align}
&\mathrm{d}_Ae=0\,,\label{sec_const_01}\\
&\mathrm{d}\tau+\langle\Omega\wed e\rangle=0\,,\label{sec_const_02}\\
&[e\wed\mathrm{d}_A\Omega]+F_A\wedge \tau=0\,,\label{sec_const_03}\\
&\langle e\wed F_A\rangle=0\,.\label{sec_const_04}
\end{align}
\end{subequations}
 Notice that, as expected from the simplified form of the field equations, these (eventually final) constraints do not depend on the Immirzi parameter.
\smallskip

\noindent \textbf{Step 4:} Before proceeding with the resolution of the equations for $Z_A$ and $Z_\Omega$, we will write down the consistency conditions that the Hamiltonian vector fields must satisfy. These are \textit{tangency} conditions that hold because the Hamiltonian vector fields must be tangent to $\mathsf{M}_1$. They are obtained by computing the Lie derivatives of the constraints along the Hamiltonian vector field. In terms of the variables $(Z_A,Z_\Omega)$ and using the constraints, they are
\begin{subequations}
\begin{align}
&[e\wed \mathrm{d}_AZ_\Omega]+[e\wed[Z_A\wed\Omega]]+\mathrm{d}_AZ_A\wedge\tau+F_A\wedge\mathrm{d}N=0\,,\label{tan_Z_2}\\
&\langle Z_\Omega\wed e\rangle=0\,,\label{tan_Z_1}\\
&[Z_A\wed e]=0\,,\label{tan_Z_3}\\
&\langle e\wed \mathrm{d}_AZ_A\rangle=0\,.\label{tan_Z_4}
\end{align}
\end{subequations}
These should be considered in conjunction with \eqref{eq: eq_Z}. The condition \eqref{tan_Z_3} implies $Z_A=0$. As a consequence, \eqref{tan_Z_4} is clearly true and \eqref{eq_Z_2} simplifies to
\[
[e\wed Z_\Omega]=NF_A\,,
\]
which is of the type that we have already encountered. Its solution is
\[
Z_\Omega^i=-N\left(\frac{e^i\wedge F_A^j}{\mathrm{vol}_e}\right)e_j\rightsquigarrow Z_\Omega\overset{\eqref{eq: star B and B^T}}{=}-N(\star_e F_A)^\top=-N(\star_e F_A)\,.
\]
Notice that this implies $\langle Z_\Omega\wed e\rangle=0$ because, from $[e\wed F_A]=0$ we get
\[
\langle Z_\Omega\wed e\rangle=-N\left(\frac{e^{[i}\wedge F_A^{j]}}{\mathrm{vol}_e}\right)e_j\wedge e_i=0\,.
\]
The previous results imply that the Hamiltonian vector fields must have the form given by (\ref{new_variables_1}) with $Z_e=0$, $Z_A=0$, $Z_\tau=0$ and $Z_\Omega=-N(\star_e F_A)$.
\smallskip

\noindent \textbf{Step 5:} At this point, the only thing left is to check that the tangency condition \eqref{tan_Z_2} holds, which can be immediately seen by plugging the values of $Z_A$ and $Z_\Omega$.
\smallskip

\noindent \textbf{Summary of the Hamiltonian formulation:}
The dynamics is given by the Hamiltonian vector field $\mathbb{X}$ with components
\begin{equation}\label{vector_fields_2}
\begin{aligned}
  X_e &= \pounds_\xi e-[\Lambda\wed e]\,, 
  && X_{e_\mathrm{t}} \text{ arbitrary}\,, \\
  X_A &= \pounds_\xi A+\mathrm{d}_A\Lambda\,, 
  && X_{A_\mathrm{t}}  \text{ arbitrary}\,, \\
  X_\Omega &= \pounds_\xi\Omega-[\Lambda\wed \Omega]+\mathrm{d}_A\psi-N(\star_e F_A)\,,
  && X_{\Omega_\mathrm{t}}  \text{ arbitrary}\,, \\
  X_\tau &= \pounds_\xi\tau+\mathrm{d}N-\langle\psi\wed e\rangle\,,
  && X_{\tau_\mathrm{t}} \text{ arbitrary}\,.
\end{aligned}
\end{equation}
It takes place on the submanifold $\mathsf{M}_1$ of the configuration space defined by
\begin{equation}\label{eq: constraint final}
\begin{aligned}
&\mathrm{d}_Ae=0\,,\\
&\mathrm{d}\tau+\langle\Omega\wed e\rangle=0\,,\\
&[e\wed\mathrm{d}_A\Omega]+F_A\wedge \tau=0\,,\\
&\langle e\wed F_A\rangle=0\,.
\end{aligned}
\end{equation}
Recall that the gauge parameters $(\Lambda,\Psi,N,\xi)$ are given by \eqref{shorthand} and that the configurations $(e_{\mathrm{t}}, e, A_{\mathrm{t}}, A, \Omega_{\mathrm{t}}, \Omega, \tau_{\mathrm{t}}, \tau)\in Q$ are subject to $e$ being non-degenerate and $N \neq 0$. These open conditions guarantee that $(e, e_{\mathrm{t}}, \tau, \tau_{\mathrm{t}})$ define a non-degenerate cotetrad. 

On $\mathsf{M}_1$, the induced presymplectic form \eqref{Omega_tan} reads  
\begin{equation}\label{presymplectic_Omega}
\omega_1= \int_\Sigma \left(\langle\dd \big(A-\gamma\Omega\big)\wwed\dd\big(\frac{1}{2\gamma}[e\wed e]\big)\rangle+\langle\dd A\wwed\dd(\tau\wedge e)\rangle\right)\,.
\end{equation}
As in the covariant picture provided by equations (\ref{EOM_1}-\ref{EOM_4}), the dynamics defined by the components of the Hamiltonian vector field can be interpreted as a combination of diffeomorphisms generated by $\xi$, internal $SU(2)$ rotations with parameter $\Lambda$ and Carrollian boosts with parameter $\psi$. In addition to this, there is non-trivial evolution defined by the two terms $-N(*_eF_A)$ and $\ddd N$ present in $X_\Omega$ and $X_\tau$, respectively. Notice that, at this point, the dynamics contains some arbitrary elements $X_{e_\mathrm{t}}$, $X_{A_\mathrm{t}}$, $X_{\Omega_\mathrm{t}}$ and $X_{\tau_{\mathrm{t}}}$. We discuss this issue in the next section.

%
%
\section{Time gauge}\label{sec_time_gauge}

An interesting and useful way to disentangle the meaning of the Hamiltonian dynamics of the system is to introduce appropriate gauge fixing conditions. This procedure can be understood in the context of reductions in principal bundles. Here we discuss a version of the time gauge customarily used to study the Holst action of GR. The idea is to impose the condition
\begin{equation}\label{time_gauge}
\tau=0
\end{equation}
This immediately tells us that $X_\tau=0$ which, in view of \eqref{vector_fields_2}, leads to
\begin{equation}\label{cond_Xtau}
0=\mathrm{d}N-\langle\psi\wed e\rangle=\mathrm{d}\tau_{\mathrm{t}}-\langle(\Omega_{\mathrm{t}}-\xi\iprod\Omega)\wed e\rangle\,.
\end{equation}
From this expression, we can solve for $\Omega_\mathrm{t}$ and write it in terms of $\{e,\Omega,\tau_\mathrm{t},e_\mathrm{t}\}$:
\begin{equation}\label{cond_Omega_t}
\Omega_{\mathrm{t}}=\xi\iprod\Omega+(\widetilde{E}\iprod \mathrm{d}\tau_\mathrm{t})/\mathsf{det}\,e\,.
\end{equation}
where we have introduced
\begin{align}
  \widetilde{E}&:=\frac{1}{2}\left(\frac{\cdot\wedge[e\wed e]}{\mathsf{vol}_0}\right)\,.\label{Etilde}
\end{align}
with $\mathsf{vol}_0$ a fiducial, non-dynamical volume form on $\Sigma$ and $\mathsf{det}\,e$ defined in Appendix \ref{subsec_topforms}. The component $X_{\Omega_\mathrm{t}}$ of the Hamiltonian vector field is no longer arbitrary and can be immediately derived from \eqref{cond_Omega_t}.

If we plug now the gauge fixing condition $\tau=0$ into the presymplectic form $\omega_1$ given by \eqref{presymplectic_Omega}, we  find
\begin{equation}\label{Omega_time_gauge}
\overline{\omega}_1=\int_\Sigma\langle\dd(A-\gamma \Omega)\wwed\dd\left(\frac{1}{2\gamma}[e\wed e]\right)\rangle\,,
\end{equation}
whereas, plugging $\tau=0$ into the constraints \eqref{eq: constraint final}, leads to
\begin{subequations}\label{sec_const_all}
\begin{align}
&\mathrm{d}_Ae=0\,,\label{sec_const_003}\\
&\langle e \wed \Omega \rangle=0\,,\label{sec_const_001}\\
&[e\wed\mathrm{d}_A\Omega]=0\,,\label{sec_const_002}\\
&\langle e\wed F_A\rangle=0\,.\label{sec_const_004}
\end{align}
\end{subequations}
The constraint \eqref{sec_const_003} has a simple interpretation: $A=\Gamma$ where $\Gamma$ is the Levi-Civita connection of the coframe $e$ given by the condition $0=\mathrm{d}_\Gamma e=\mathrm{d}e+[\Gamma \wed e]$.

The symplectic form \eqref{Omega_time_gauge} can be simplified using $0=\dd(\mathrm{d}_\Gamma e)=\mathrm{d}_\Gamma \dd e+[\dd\Gamma\wed e]$ and, hence,
\begin{align*}
\int_\Sigma\langle\dd\Gamma\wwed\dd[e\wed e]\rangle&=2\int_\Sigma\langle\dd\Gamma\wwed[\dd e\wed e]\rangle=2\int_\Sigma\langle[\dd\Gamma\wed e]\wwed\dd e\rangle=-2\int_\Sigma \langle\mathrm{d}_\Gamma\dd e\wwed\dd e\rangle\\
&=-\int_\Sigma \mathrm{d}\langle\dd e\wwed \dd e\rangle=0,
\end{align*}
because $\Sigma$ has no boundary. Using this once we plug $A=\Gamma$ into \eqref{Omega_time_gauge} leads to 
\begin{equation}\label{Omega_time_gauge_3}
\overline{\omega}_1=\int_\Sigma\langle\dd \left(\frac{1}{2}[e\wed e]\right)\wwed\dd\Omega\rangle\,.
\end{equation}
or, in terms of $\widetilde{E}$
\begin{equation}\label{Omega_time_gauge_4}
\widehat{\omega}_1=\int_\Sigma\langle\dd\widetilde{E}\wwed\dd\Omega\rangle\mathsf{vol}_0\,.
\end{equation}
The constraints \eqref{sec_const_all} in terms of the new variables are just simply obtained by writing $e$ in terms of $\widetilde{E}$ and setting  $A=\Gamma$:
\begin{subequations}
\begin{align}
&[\widetilde{E}\iproddot\Omega]=0\,,\label{sec_const_000001}\\
&\langle\widetilde{E}\iproddot\mathrm{d}_\Gamma\Omega\rangle=0\,,\label{sec_const_000002}\\
&\langle\widetilde{E}\iproddot[\widetilde{E}\iproddot F_\Gamma]\rangle=0\,.\label{sec_const_000004}
\end{align}
\end{subequations}
Finally, we have to find the form of the components of the Hamiltonian vector fields that give the dynamics of $\Omega$ and $\widetilde{E}$. In the case of $X_\Omega$ it is immediate to get 
\begin{equation}\label{X_Omega}
X_\Omega=\pounds_\xi\Omega-[\Lambda\wed \Omega]+\mathrm{d}_\Gamma\big((\widetilde{E}\iprod\mathrm{d}\tau_{\mathrm{t}})/\mathsf{det}\,e\big)-\tau_{\mathrm{t}}\ast_eF_{\Gamma}(\widetilde{E})\,.
\end{equation}
In order to find $X_{\widetilde{E}}$, we use $2\widetilde{E}\iprod\mathsf{vol}_0=[e\wed e]$ to get
\begin{align*}
X_{\widetilde{E}}\iprod\mathsf{vol}_0&=[X_e\wed e]=[(\pounds_\xi e-[\Lambda\wed e])\wed e]=\pounds_\xi\left(\frac{1}{2}[e\wed e]\right)-\frac{1}{2}[\Lambda\wed[e\wed e]]\\
&=\pounds_\xi(\widetilde{E}\iprod\mathsf{vol}_0)+[\Lambda\wed(\widetilde{E}\iprod \mathsf{vol}_0)]=\Big(\pounds_\xi\widetilde{E}+(\mathsf{div}_0\xi)\widetilde{E}+[\Lambda,\widetilde{E}]\Big)\iprod\mathsf{vol}_0
\end{align*}
where we have used the Jacobi identity (see Appendix \ref{subsec_Lie}). The preceding result immediately implies
\begin{equation}\label{X_Etilde}
X_{\widetilde{E}}=\pounds_\xi\widetilde{E}+(\mathsf{div}_0\xi)\widetilde{E}-[\Lambda,\widetilde{E}]\,.
\end{equation}
Notice that, at this point, there are still a number of dynamical variables with arbitrary dynamics such as $e_\mathrm{t}$, $A_\mathrm{t}$ and $\tau_\mathrm{t}$ that we must take care of. Before doing that, some comments are in order. First, the final form of the constraints is very similar to the constraints from which the Ashtekar formulation for GR is derived by the standard canonical transformation. In fact, the only difference is the scalar constraint \eqref{sec_const_000004}. The usual procedure would then allow us to find an Ashtekar formulation for this model that will only differ in the form of the Hamiltonian constraint. In fact, the form of \eqref{Omega_time_gauge} directly suggests to introduce the Ashtekar connection $\Gamma-\gamma\Omega$ which would be canonically conjugate to $\gamma^{-1}\widetilde{E}$.
Second, the Hamiltonian constraint found here coincides (if multiplied by an appropriate factor) with the extra term that must be added to the Euclidean Hamiltonian constraint in the Ashtekar formulation in order to describe Lorentzian GR. Finally, the constraints \eqref{sec_const_000001}-\eqref{sec_const_000004} coincide with the ones given in \cite{Sengupta:2022rbd} after using the time gauge. Notice, however, that the approach that we have followed here differs in several significant ways from the one of \cite{Sengupta:2022rbd}: (i) we use the most general Carrollian action that can be derived from the Holst action as shown in \cite{Figueroa-OFarrill:2022mcy}, (ii) we use differential forms as the basic variables. This makes the Hamiltonian analysis much simpler and allows us to employ the powerful machinery of the exterior calculus. Notice also that in our approach we do not need to introduce the (awkward) quadratic constraint used in the literature to get the Hamiltonian formulation for the Holst action \cite{BarroseSa:2000vx}.

 As written above, the expressions for the Hamiltonian vector fields are actually quite complicated because they are expressed in terms of $\xi$, $\Lambda$, $\psi$ and $\tau_\mathrm{t}$. Notice, for instance, that $\xi=\langle\widetilde{E}\iprod e_\mathrm{t}\rangle/(\mathsf{det} \, e)$ with $(\mathsf{det} \, e)$ written in terms of $\widetilde{E}$. 
 This issue can be solved by imposing additional gauge fixing conditions. In particular, by choosing a smooth vector field $\hat{\xi}\in\mathfrak{X}(\Sigma)$ on the spatial manifold $\Sigma$ and requiring that $\hat{\xi}\iprod e=e_\mathrm{t}$, we see that the components \eqref{X_Omega} and \eqref{X_Etilde} of the Hamiltonian vector field become
\begin{align}
&X_\Omega=\pounds_{\hat{\xi}}\Omega-[\Lambda\wed \Omega]+\mathrm{d}_\Gamma\big((\widetilde{E}\iprod\mathrm{d}\tau_{\mathrm{t}})/\mathsf{det}\,e\big)-\tau_{\mathrm{t}}\ast_eF_{\Gamma}(\widetilde{E})\label{X_Omega_hat}\\
&X_{\widetilde{E}}=\pounds_{\hat{\xi}}\widetilde{E}+(\mathsf{div}_0\hat{\xi})\widetilde{E}-[\Lambda,\widetilde{E}]\,.\label{X_Etilde_hat}
\end{align}
where there is still some arbitrariness as we have not fixed $\Lambda$ nor $\tau_{\mathrm{t}}$ yet. It is important to keep in mind that the condition $\hat{\xi}\iprod e=e_\mathrm{t}$ gives $e_{\mathrm{t}}$ as a function of $e$ and tells us also about its dynamics because 
\[
X_{e_{\mathrm{t}}}=\hat{\xi}\iprod(\pounds_{\hat{\xi}}e-[\Lambda\wed e])=\pounds_{\hat{\xi}}e_{\mathrm{t}}-[\Lambda\wed e_{\mathrm{t}}]\,.
\]
In order to reconstruct the original dynamical objects appearing in the action we need to know $A_\mathrm{t}$ and $\tau_{\mathrm{t}}$. An important issue to have in mind is that, within the time independent formalism that we are using here, this is not simply a matter of choosing the remaining arbitrary fields as any time function we wish. Although it is true that there is some arbitrariness in them, the evolution of an object such as $A_{\mathrm{t}}$ is given by an equation of the type
\[
\dot{A}_{\mathrm{t}}=X_{A_{\mathrm{t}}}
\] 
where $X_{A_{\mathrm{t}}}$ may be any sufficiently regular function \textit{in phase space}. A possible and consistent choice would be a smooth, nowhere-vanishing function $\hat{\tau}_{\mathrm{t}}\in \Omega^0(\Sigma)$ and $\hat{A}_{\mathrm{t}}\in \Omega^0(\Sigma,\mathfrak{su}(2))$. In such case, the corresponding components of the Hamiltonian vector fields vanish ($X_{\tau_{\mathrm{t}}}=0$ and $X_{A_{\mathrm{t}}}=0$). Another option for $\hat{A}_{\mathrm{t}}$ would be to set $\Lambda=0$ and, hence, the internal rotations disappear from the dynamics. Notice that it is only after this gauge fixing process that we get a simple interpretation of the dynamics defined by \eqref{X_Omega_hat} and \eqref{X_Etilde_hat} in terms of the gauge parameters $(\hat{\xi},\hat{\Lambda},\hat{\tau}_{\mathrm{t}},\hat{\Omega}_{\mathrm{t}})$. 

We end this section by emphasizing, again, that the presymplectic form \eqref{Omega_time_gauge} immediately suggests how to find a formulation for this model in the Ashtekar phase space for GR. As mentioned above, the Ashtekar variables can be directly read from $\bar{\omega}_1$. The computations leading to the Ashtekar formulation for this model are the standard ones, so we will not repeat them here (a simple account adapted to the language we use here can be found in \cite{BarberoG:2020tit}). Notice that this formulation can only be found because the action considered in the paper contains an extra term involving the Immirzi parameter.

%
%
\section{Conclusions and comments}\label{sec_conclusions}

In this paper, we have studied the Hamiltonian formulation for the most general Carroll invariant action that can be obtained from the Holst action. For this purpose, we have relied on the GNH approach. When this method (or the essentially equivalent Dirac approach) is used, one of the key steps is to find the explicit form of the Hamiltonian vector fields that define the dynamics and check its consistency, in the sense that they must preserve the constraints. The resolution of the linear inhomogeneous equations for the components of these fields is usually hard, in fact, a brute force approach usually leads to unwieldy forms for their solutions that are hard to simplify by making use of the constraints. This complicates the interpretation of the dynamics.

A neat way to circumvent these difficulties can be found by looking at the geometric underpinnings of the model and the meaning of the Carrollian symmetry. A convenient setting for this is provided by Cartan geometry and the tools that it provides. By understanding how the symmetry works and the meaning of relevant geometric objects such as the torsion and the curvature, it is actually possible to guess some of the terms that are expected to appear in the Hamiltonian vector fields \eqref{new_variables_1} and rewrite the relevant equations in terms of variables that simplify their solution in a significant way as shown in the paper. In particular, the expected gauge transformations originating in the Carrollian symmetry appear in a clear form.

We have also shown that it is possible to identify a dynamically determined congruence of curves such that (part of) the dynamics can be understood in simple terms because Lie dragging solutions to the field equations along these curves produces simple internal gauge transformations. This is similar to the behavior of the Husain-Kucha\v{r} model as discussed in \cite{BarberoG:2025rxm}. In any case, for the magnetic Carrollian model that the action considered in the paper describes, the evolution \textit{does not} simply reduce to internal rotations or Carrollian boosts, in particular for the $\Omega$ field. It is noteworthy to point out that the torsion vanishes.

From the perspective of the classical dynamics, our description of the system is complete: once the initial data for the basic fields is prescribed in such a way that they satisfy all the constraints, their evolution is dictated by the Hamiltonian vector field given there. On the other hand, if one is interested in the Dirac quantization of the model, it is necessary to find a symplectic formulation such that the symplectic form takes the simplest possible form (in practice, it must have the canonical $\mathrm{d}q\wedge \mathrm{d}p$ form). This requires some extra work. For the model discussed here, the time gauge that is successfully used for this purpose in other settings (for instance, when dealing with the Holst action) can be used to get a simple formulation as shown in Section \ref{time_gauge}. Furthermore, by following the same steps leading to the Ashtekar formulation for GR, it is possible to find an Ashtekar formulation for the model. This may help quantize it and illuminate the harder problem of quantizing GR.

\section*{Acknowledgments}

F. Barbero acknowledges the support of the WOST, WithOut SpaceTime project (https://withoutspacetime.org), supported by Grant ID\# 63683 from the John Templeton Foundation (JTF). The opinions expressed in this work are those of the author(s) and do not necessarily reflect the views of the John Templeton Foundation. J. Margalef-Bentabol is supported by NSERC grant DGECR-2026-00316.

\begin{appendices}

%
%
%

\section{Some useful mathematical results}\label{sec_math_results}

\subsection{Algebraic foundations and vector-valued forms}\label{subsec_Lie} 

Let $\Omega^p(\Sigma, \mathfrak{su}(2))$ denote the space of $\mathfrak{su}(2)$-valued $p$-forms on $\Sigma$. Given a basis $\{J_i\}_{i=1}^3$ of $\mathfrak{su}(2)$ satisfying $[J_i, J_j] = \tensor{\epsilon}{_i_j^k} J_k$, any $\alpha \in \Omega^p(\Sigma, \mathfrak{su}(2))$ expands as $\alpha = \alpha^i J_i$, where $\alpha^i \in \Omega^p(\Sigma)$.

We equip $\mathfrak{su}(2)$ with an inner product $\langle \cdot, \cdot \rangle := -\frac{1}{2} K(\cdot, \cdot)$, where $K$ is the standard Cartan-Killing metric. In our chosen basis, $\langle J_i, J_j \rangle = \delta_{ij}$. The exterior product and the Lie algebra bracket extend naturally to $\mathfrak{su}(2)$-valued forms:
\begin{align}
    \langle \alpha \wed \beta \rangle &:= \alpha^i \wedge \beta^j \langle J_i, J_j \rangle=\alpha^i\wedge \beta_i \in \Omega^{a+b}(\Sigma) \,, \\
    [\alpha \wed \beta] &:= \alpha^i \wedge \beta^j [J_i, J_j]= \alpha^i \wedge \beta^j \tensor{\epsilon}{_i_j^k}J_k \in \Omega^{a+b}(\Sigma, \mathfrak{su}(2)) \,.
\end{align}
Let $A \in \Omega^1(\Sigma, \mathfrak{su}(2))$ be a connection 1-form. The associated gauge-covariant exterior derivative $\mathrm{d}_A$ and its curvature $F_A \in \Omega^2(\Sigma, \mathfrak{su}(2))$ are defined by
\begin{align}
    \mathrm{d}_A \beta &:= \mathrm{d}\beta + [A \wed \beta] \,, \\
    F_A &:= \mathrm{d}A + \frac{1}{2}[A \wed A] \,.
\end{align}
Finally, for any $\alpha \in \Omega^a(\Sigma, \mathfrak{su}(2))$, $\beta \in \Omega^b(\Sigma, \mathfrak{su}(2))$, $\gamma \in \Omega^c(\Sigma, \mathfrak{su}(2))$, and $\xi \in \mathfrak{X}(\Sigma)$, the following identities hold by virtue of the graded Jacobi identity and the invariance of the inner product:
\begin{enumerate}
    \item $\langle \alpha \wed \beta \rangle = (-1)^{ab} \langle \beta \wed \alpha \rangle$
    \item $[\alpha \wed \beta] = (-1)^{ab+1} [\beta \wed \alpha]$
    \item $\langle [\alpha \wed \beta] \wed \gamma \rangle = \langle \alpha \wed [\beta \wed \gamma] \rangle$
    \item $\mathrm{d}_A [\alpha \wed \beta] = [\mathrm{d}_A \alpha \wed \beta] + (-1)^a [\alpha \wed \mathrm{d}_A \beta]$
    \item $\mathrm{d}_A \langle \alpha \wed \beta \rangle = \langle \mathrm{d}_A \alpha \wed \beta \rangle + (-1)^a \langle \alpha \wed \mathrm{d}_A \beta \rangle$
    \item $\xi \iprod [\alpha \wed \beta] = [(\xi \iprod \alpha) \wed \beta] + (-1)^a [\alpha \wed (\xi \iprod \beta)]$
    \item $\mathrm{d}_A^2 \alpha = [F_A \wed \alpha]$ \quad and \quad $\mathrm{d}_A F_A = 0$ (Ricci and Bianchi identities)
    \item $[\alpha\wed[\beta\wed\gamma]]=(-1)^{bc}\langle\alpha\wed\gamma\rangle\wedge\beta-\langle\alpha\wed\beta\rangle\wedge\gamma$
\end{enumerate}

\subsection{Top forms and frame dualities}\label{subsec_topforms}

Let $\Omega^3(\Sigma)$ be the space of top-forms on $\Sigma$ and consider a reference volume form $\mathsf{vol}_0 \in \Omega^3(\Sigma)$. The quotient of any top-form $\mu \in \Omega^3(\Sigma)$ with respect to $\mathsf{vol}_0$ is uniquely defined as the smooth scalar function $\left(\frac{\mu}{\mathsf{vol}_0}\right) = f$ satisfying $\mu = f \mathsf{vol}_0$. 

Let $e^i \in \Omega^1(\Sigma)$, $i=1,2,3$,  be a non-degenerate spatial cotriad. The dynamical spatial volume form $\mathsf{vol}_e$ and its corresponding density factor $\mathsf{det} \, e$ are given by
\begin{equation}
    \mathsf{vol}_e := \frac{1}{3!}\epsilon_{ijk} e^i \wedge e^j \wedge e^k \,, \quad \mathsf{det}\, e := \left(\frac{\mathsf{vol}_e}{\mathsf{vol}_0}\right) \,.
\end{equation}
Let $H := \frac{1}{2}[e\wed e] \in \Omega^2(\Sigma, \mathfrak{su}(2))$. We introduce the dual vector fields $\vec{E}^i \in \mathfrak{X}(\Sigma)$ and their densitized counterpart $\widetilde{E}^i$ relative to $\mathsf{vol}_0$ via the relations:
\begin{equation}
    \vec{E}^i \iprod \, \mathsf{vol}_e := H^i \,, \quad \widetilde{E}^i \iprod \, \mathsf{vol}_0 := H^i \,.
\end{equation}
These definitions yield the normalization $\vec{E}^i \iprod \, e_j = \delta^i_j$ and imply $\widetilde{E}^i = (\mathsf{det} \, e)\vec{E}^i$. By using these equations repeteadly, one gets:
\begin{align}
    e^i &= -\frac{1}{2\,\mathsf{det} \, e}\epsilon^{ijk}\widetilde{E}_j \iprod \, \widetilde{E}_k \iprod \, \mathsf{vol}_0 \,, \\
    (\mathsf{det} \, e)^2 &= -\frac{1}{3!}\epsilon^{ijk}\widetilde{E}_i \iprod \, \widetilde{E}_j \iprod \, \widetilde{E}_k \iprod \, \mathsf{vol}_0 \,.
\end{align}

\subsection{Simplification of the field equations}\label{subsec_simplif}

We establish here the equivalence between the field equations obtained directly by varying the action and their simplified forms.

\medskip

\begin{proposition}
The system of equations
\begin{subequations}\label{Equation_comp}
\begin{align}
    &\mathrm{d}_{\Aa} [\ee \wed \ee] = 0 \,, \label{Equation_comp_1}\\
    &2\mathrm{d}_{\Aa} (\uptau \wed \ee) - [\Upomega \wed [\ee \wed \ee]] = 0 \,, \label{Equation_comp_2}
\end{align}
\end{subequations}
is equivalent to
\begin{subequations}\label{Equation_simp}
\begin{align}
    &\mathrm{d}_{\Aa} \ee = 0 \,, \label{Equation_simp_1}\\
    &\mathrm{d}\uptau + \langle \Upomega\wed \ee \rangle = 0 \,. \label{Equation_simp_2}
\end{align}
\end{subequations}
\end{proposition}
\begin{proof}
It is straightforward to check that \eqref{Equation_simp} implies \eqref{Equation_comp}. To prove the converse, we span $\mathrm{d}_\Aa  \ee$ and $\mathrm{d}\uptau$ in the basis components provided by the coframe $\{\uptau, \ee^i\}$:
\begin{subequations}
\begin{align}
    \mathrm{d}_\Aa \ee^i &= \tensor{\alpha}{^i_j} \ee^j \wedge \uptau + \tensor{\beta}{^i_j} \tensor{\epsilon}{^j_p_q} \ee^p \wedge \ee^q \,, \label{decomp1} \\
    \mathrm{d}\uptau &= s_i \ee^i\wedge \uptau + m^i \epsilon_{ijk} \ee^j \wedge \ee^k \,, \label{decomp2} \\
    \Upomega^i &= \tensor{\Upomega}{^i_j} \ee^j +\phi^i\uptau\,. \label{decomp3}
\end{align}
\end{subequations}
Substituting \eqref{decomp1} into \eqref{Equation_comp_1} we get $\epsilon_{ijk} (\mathrm{d}_{\Aa} \ee^j) \wedge \ee^k = 0$. Wedging this expression with $\uptau$ yields $\epsilon_{ijk}\beta^{jk} = 0$, while wedging with $\ee^r$ yields $\tensor{\alpha}{^r_i} - \tensor{\delta}{^r_i}\tensor{\alpha}{^j_j} = 0$, which implies $\tensor{\alpha}{^i_j} = 0$. 

Next, substituting these conditions into \eqref{Equation_comp_2} and exploiting the algebraic identity $-[\Upomega \wed [\ee \wed \ee]] = 2\langle \Upomega \wed \ee \rangle \wedge \ee$, we get:
\begin{equation}
    -s_j \uptau \wedge \ee^j \wedge \ee^i + m_j\epsilon^{jk\ell} \ee_k \wedge \ee_\ell \wedge \ee^i - \tensor{\beta}{^i_j}\epsilon^{jk\ell}\uptau \wedge \ee_k \wedge \ee_\ell - \Upomega^{jk}\ee_j\wedge \ee_k\wedge \ee^i+\phi^j\uptau\wedge \ee_j\wedge \ee^i = 0 \,.
\end{equation}
Wedging from the left with $\uptau$ gives $2m_i - \epsilon_{ijk}\Upomega^{jk} = 0$, while wedging from the right with $\ee^p$ we get $\epsilon^{ipk}s_k + 2\beta^{ip}-\epsilon^{ipk}\phi_k = 0$. Contracting this final relation with $\epsilon_{ip\ell}$ and using $\epsilon^{ijk}\beta_{jk}=0$ yields $s_i = \phi_i$, which implies $\beta_{ij} = 0$. Plugging $\alpha_{ij}=0$, $\beta_{ij}=0$, $s_i=\phi_i$ and $2m_i=\varepsilon_{ijk}\Omega^{jk}$ into \eqref{decomp1} and \eqref{decomp2} we immediately arrive at \eqref{Equation_simp_1} and \eqref{Equation_simp_2}. 
\end{proof} 
\begin{proposition}
    If $\mathrm{d}_{\Aa} \ee = 0$ holds, the field equation
\begin{equation}\label{eef}
    [\ee \wed (\kappa F_{\Aa} - \gamma \mathrm{d}_{\Aa} \Upomega)] = \gamma F_{\Aa} \wedge \tau
\end{equation}
is equivalent to
\begin{equation}\label{eee}
    F_{\Aa} \wedge \uptau + [\ee \wed \mathrm{d}_{\Aa} \Upomega] = 0 \,.
\end{equation}
\end{proposition} 
\begin{proof}
    Since $\mathrm{d}_{\Aa} \ee = 0$ implies the $[\ee \wed F_{\Aa}] = 0$, the term proportional to $\kappa$ vanishes. As $\gamma \neq 0$ we arrive at \eqref{eee}. 
\end{proof}

\subsection{Some useful lemmas on frame contractions}\label{subsec_lemma}

\begin{lemma}\label{lemma_[e,B]=0}
    Let $\mathsf{B} \in \Omega^2(M, \mathfrak{su}(2))$ on a 4-dimensional manifold $M$. Let $(\mathsf{u}, \mathsf{E}_i)$ be the frame field dual to the spacetime coframe $(\uptau, \ee^i)$, such that $\mathsf{u} \iprod \, \uptau = 1$ and $\mathsf{u} \iprod \, \ee^i = 0$. If $\mathsf{B}$ satisfies $[\ee \wed \mathsf{B}] = 0$ then 
\begin{equation}
\mathsf{u} \iprod \, \mathsf{B} = 0\,,\qquad  \mathsf{B}=\frac{1}{2} \tensor{\mathsf{B}}{_i_j^k} \ee^i\wedge \ee^j J_k
    \end{equation}
 and (see \eqref{eq: star B and B^T} for the notation)
\begin{equation}(*_\ee \mathsf{B}) =(*_\ee \mathsf{B})^\top\,.\end{equation}     
\end{lemma}

\begin{proof}
 Let us introduce the 4-dimensional spacetime volume form $\mathsf{vol} := \tau \wedge \mathsf{vol}_e$. Since $\mathsf{u} \iprod \, \ee^i = 0$, contracting the condition $[\ee \wed \mathsf{B}] = 0$ with the vector field $\mathsf{u}$ yields $[\ee \wed (\mathsf{u} \iprod \, \mathsf{B})] = 0$. Expanding the 1-form $\mathsf{u} \iprod \, \mathsf{B}= (\mathsf{u} \iprod \, \mathsf{B})^k J_k$  on the local coframe basis gives:
\begin{equation}
    (\mathsf{u} \iprod \mathsf{B})^k = b^k \tau + \tensor{b}{^k_\ell} e^\ell \,.
\end{equation}
Contracting a second time yields $\mathsf{u} \iprod \, (\mathsf{u} \iprod \, \mathsf{B})^k = 0$, which sets $b^k = 0$. The vanishing bracket condition then simplifies to $\epsilon_{ijk} \tensor{b}{^k_\ell} \ee^j \wedge  \ee^\ell = 0$. Taking the wedge product of this equation with $\uptau \wedge \ee^m$ turns the constraint into:
\begin{equation}
    \epsilon_{ijk} \tensor{b}{^k_\ell} \, \uptau \wedge \ee^m \wedge \ee^j \wedge \ee^\ell = 0 \Rightarrow \epsilon_{ijk}\varepsilon^{mj\ell} \tensor{b}{^k_\ell} = 0 \,.
\end{equation}
Using the standard identity $\epsilon_{ijk}\epsilon^{mj\ell} = \delta_i^\ell \delta_k^m - \delta_i^m \delta_k^\ell$, we find $\delta_i^m \tensor{b}{^k_k} - \tensor{b}{^m_i} = 0$. Tracing this relation implies $\tensor{b}{^m_i} = 0$, which proves that $\mathsf{u} \iprod \, \mathsf{B} = 0$. This condition implies that
\[
\mathsf{B}=\frac{1}{2} \tensor{\mathsf{B}}{_i_j^k} \ee^i\wedge \ee^j J_k\,.
\]
From  
\[
0=[\ee\wed \mathsf{B}]= \frac{1}{2} \tensor{\mathsf{B}}{_j_k^\ell} \ee^i\wedge\ee^j\wedge \ee^k[J_i,J_\ell]=\frac{1}{2}\epsilon^{ijk}\epsilon_{j\ell m}\tensor{\mathsf{B}}{_j_k^\ell} \ee^1\wedge \ee^2\wedge \ee^3 J_m
\]
it follows that
\[
 \star_\ee \mathsf{B}=\frac{1}{2}\tensor{\epsilon}{_i^k^l}\tensor{\mathsf{B}}{_k _l ^j} (\ee^i J_j)= \frac{1}{2}\tensor{\epsilon}{^i^k^l}\tensor{\mathsf{B}}{_k _l _j} (\ee^j J_i)=
 (\star_\ee \mathsf{B})^\top\,.
\]
\end{proof}
By a similar argument, we can prove the following
\begin{lemma}\label{lemma_[e,X]=C} The equation
\[
[\ee\wed \mathsf{X}] = \mathsf{C},
\]
where \(\mathsf{X}\in\Omega^1(M,\mathfrak{su}(2))\) satisfies
\(\mathrm{u}\iprod \mathsf{X}=0\), admits a solution if and only if
\(\mathrm{u}\iprod \mathsf{C}=0\). In this case, the solution is unique and is given by
\[
\mathsf{X}
=\frac{1}{2}\bigl(*_\ee \langle \ee\wed \mathsf{C}\rangle\bigr)\ee
-(*_\ee \mathsf{C})^\top.
\]

\end{lemma}

\section{Spacetime diffeomorphisms and  gauge invariance}\label{app_juan4}

In this section, we explicitly compute the transformations of the spacetime fields $(\Aa, \Upomega, \ee, \uptau)$ under a spacetime vector field decomposed relative to a spatial hypersurface. Let $\jmath: \Sigma \hookrightarrow M$ be the embedding map of the spatial hypersurface, and consider a generic vector field $\upxi$ over $\jmath$ which can be decomposed as 
\begin{equation}
    \upxi = \upalpha \partial_{t} + \jmath_{*}\beta \,,
\end{equation}
where $\upalpha$ is a lapse-like scalar function, $\partial_t$ is the time-evolution vector field transverse to $\jmath(\Sigma)$, and $\beta \in \mathfrak{X}(\Sigma)$ is the spatial shift vector. Under the structural decomposition of the fields, the time derivatives are governed by combinations of spatial Lie derivatives along a reference vector field $\xi \in \mathfrak{X}(\Sigma)$ and internal gauge transformations. 

By expanding the spacetime Lie derivatives and using the pullback map $\jmath^*$, we isolate the spatial configurations from their transverse/temporal counterparts.

\bigskip

\noindent \textbf{Transformation of the connection $\Aa$:}
    \begin{align*}
        \jmath^{*}(\pounds_{\upxi}\Aa) &=\jmath^{*}((\upalpha \partial_t)\iprod\mathrm{d}\Aa+\mathrm{d}((\upalpha \partial_t)\iprod\Aa) +\jmath^{*}(\pounds_{\jmath_*\beta}\Aa)\\
        &= \alpha\dot{A}  + A_{\mathrm{t}} \mathrm{d}\alpha  + \pounds_{\beta}A\overset{\eqref{vector_fields_2}}{=} \\
        &= \alpha(\pounds_{\xi}A + \mathrm{d}_{A}\Lambda) + \pounds_{\beta}A + A_{\mathrm{t}}\mathrm{d}\alpha \\
        &= \pounds_{\alpha \xi + \beta}A + \mathrm{d}_{A}(\alpha \Lambda)\overset{\eqref{vector_fields_2}}{=} \dot{A}|_{\xi\mapsto \alpha \xi+\beta,\Lambda\mapsto \alpha\Lambda} \,.
    \end{align*}
    
\noindent \textbf{Transformation of the cotriad $\mathbf{e}$:}
    \begin{align*}
        \jmath^{*}(\pounds_{\upxi}\ee) &= \alpha \dot{e} + e_{\mathrm{t}}\mathrm{d}\alpha + \pounds_{\beta}e \overset{\eqref{vector_fields_2}}{=}\\
        &= \pounds_{\alpha \xi + \beta}e - [(\alpha \Lambda) \wed e] \overset{\eqref{vector_fields_2}}{=} \dot{e}|_{\xi\mapsto \alpha \xi+\beta,\Lambda\mapsto \alpha\Lambda} \,.
    \end{align*}
    
\noindent\textbf{Transformation of the connection component $\Upomega$:}
    \begin{align*}
        \jmath^{*}(\pounds_{\upxi}\Upomega) &= \alpha \dot{\Omega}  + \Omega_{\mathrm{t}}\mathrm{d}\alpha + \pounds_{\beta}\Omega\overset{\eqref{vector_fields_2}}{=}\\ 
        &= \pounds_{\alpha \xi + \beta} \Omega - [(\alpha \Lambda) \wed \Omega ] - \alpha N (*_e  F_A ) + \mathrm{d}_{A}(\alpha \psi) \overset{\eqref{vector_fields_2}}{=} \dot{\Omega}|_{\xi\mapsto \alpha \xi+\beta,\Lambda\mapsto \alpha\Lambda,\psi\mapsto \alpha\psi,N\mapsto \alpha N} \,.
    \end{align*}
    
\noindent \textbf{Transformation of the clock field $\uptau$:}
    \begin{align*}
        \jmath^{*}(\pounds_{\upxi} \uptau ) &= \alpha \dot{\tau} + \pounds_{\beta}\tau  + \tau_{\mathrm{t}}\mathrm{d}\alpha \overset{\eqref{vector_fields_2}}{=}\\ 
        &= \pounds_{\alpha \xi + \beta} \tau + \mathrm{d}(\alpha N)  - \langle (\alpha \psi) \wed e  \rangle  \overset{\eqref{vector_fields_2}}{=} \dot{\tau}|_{\xi\mapsto \alpha \xi+\beta,\psi\mapsto \alpha\psi,N\mapsto \alpha N} \,.
    \end{align*}
    
Geometrically, this systematic reduction demonstrates that the full 4-dimensional spacetime symmetries project consistently down to the spatial hypersurface. The resulting field variations are equivalent to a combined spatial diffeomorphism along the modified shifted vector field $\alpha\xi + \beta$, supplemented by local internal gauge rotations parameterized by the fields $\alpha\Lambda$, $\alpha\psi$, and $\alpha N$.

\end{appendices}

%
%

\end{document}